\documentclass[11pt]{article}
\usepackage[a4paper, margin=1in]{geometry}
\usepackage{authblk}
\usepackage[UKenglish]{babel}
\usepackage{amsmath}
\usepackage{amsfonts}
\usepackage{amsthm}
\usepackage{amssymb}
\usepackage{mathtools}
\mathtoolsset{centercolon}
\usepackage{thmtools, thm-restate}
\usepackage{stmaryrd}
\usepackage{lineno}
\usepackage{caption}

\usepackage{url}    
\usepackage{hyperref}
\usepackage{orcidlink}

\usepackage{tikz}
\usetikzlibrary{decorations.pathreplacing}

\usepackage[backend=biber, style=numeric, sorting=nyt]{biblatex}
\usepackage[capitalise,noabbrev]{cleveref}

\DeclareMathOperator{\supp}{sup}

\DeclareMathOperator{\MOD}{MOD}
\DeclareMathOperator{\AND}{AND}

\DeclareMathOperator{\maxSup}{maxSup}

\newcommand{\Aut}{\mathbf{Aut}}
\newcommand{\Sym}{\mathbf{Sym}}
\newcommand{\Stab}{\mathbf{Stab}}
\newcommand{\StabP}{\mathbf{Stab}^{\bullet}}
\newcommand{\Orb}{\mathbf{Orb}}

\newcommand{\bbN}{\mathbb{N}}
\newcommand{\bbZ}{\mathbb{Z}}

\newcommand{\Bb}{{\mathcal B}}

\newcommand{\Oo}{{\mathcal O}}

\newcommand{\Tt}{{\mathcal T}}

\newcommand{\Xx}{{\mathcal X}}

\usepackage{csquotes}

\usepackage{comment}

\usepackage{todonotes}

\usepackage{enumitem}

\theoremstyle{definition}
\newtheorem{definition}{Definition}[section]

\theoremstyle{plain}
\newtheorem{lemma}[definition]{Lemma}
\newtheorem{corollary}[definition]{Corollary}
\newtheorem{theorem}[definition]{Theorem}

\Crefname{fact}{Fact}{Facts}

\theoremstyle{remark}
\newtheorem{claim}{Claim}[lemma]

\Crefname{claim}{Claim}{Claims}
\newenvironment{claimproof}[1][Proof of Claim]{\begin{proof}[#1] }{ \end{proof}}
\setlist[enumerate, 1]{font=\upshape, noitemsep, nolistsep}
\setlist[enumerate, 2]{font=\upshape, noitemsep, nolistsep}
\setlist[itemize, 1]{noitemsep, nolistsep,font=\upshape}
\setlist[itemize, 2]{noitemsep, nolistsep,font=\upshape}

\theoremstyle{remark}

\Crefname{claim}{Claim}{Claims}

\newcommand{\maxorb}{\operatorname{maxOrb}}
\newcommand{\maxsup}{\operatorname{maxSup}}

\tikzset{
	vertex/.style={draw,circle,fill=gray},
	every node/.style={anchor=center},
	lbl/.style={color=lightgray}
}

\title{Optimal Lower Bounds for Symmetric Modular Circuits}
\author{Benedikt Pago \orcidlink{0000-0001-6377-1230} \thanks{An extended abstract of this work appeared at ICALP 2026 \cite{pago_ICALP2026}}}
\affil{Department of Computer Science and Technology, University of Cambridge, UK}

\renewcommand{\phi}{\varphi}
\renewcommand{\epsilon}{\varepsilon}

\DeclareMathOperator{\CC}{CC}
\DeclareMathOperator{\ACC}{ACC}
\DeclareMathOperator{\AC}{AC}
\DeclareMathOperator{\MAJ}{MAJ}
\DeclareMathOperator{\OR}{OR}

\DeclareMathOperator{\StabIn}{\mathbf{Stab}_{in}}

\begin{document}

\maketitle

\begin{abstract}
	A notorious open question in circuit complexity is whether Boolean operations of arbitrary arity can efficiently be expressed using modular counting gates only. 
	H{\aa}stad's celebrated switching lemma yields exponential lower bounds for the dual problem -- realising modular arithmetic with Boolean gates -- but, a similar lower bound for modular circuits computing the Boolean AND function has remained elusive for almost 30 years. 
	
	We solve this problem for the restricted model of symmetric circuits: We consider $\MOD_m$-circuits of arbitrary depth, and for an arbitrary modulus $m \in \bbN$, and obtain subexponential lower bounds for computing the $n$-ary Boolean AND function, under the assumption that the circuits are syntactically symmetric under all permutations of their $n$ input gates. This lower bound is matched precisely by a construction due to (Idziak, Kawa\l ek, Krzaczkowski, LICS'22), leading to the surprising conclusion that the optimal symmetric circuit size is already achieved with depth $2$. 
	
	Motivated by another construction from (LICS'22), which achieves smaller size at the cost of greater depth, we also prove tight size lower bounds for circuits with a more liberal notion of symmetry characterised by a nested block structure on the input variables.
\end{abstract}

\section{Introduction}
There are many long-standing open questions in circuit complexity that are surprisingly simple in their formulation, yet no solution to them has been found in decades.
Among them there is the following: Can the $n$-ary Boolean $\AND_n$-function be represented with a polynomial-size circuit of depth two using only $\MOD_6$-gates?
A $\MOD_6$-gate is a Boolean gate which takes an arbitrary number of inputs and returns $1$ if and only if their sum modulo $6$ belongs to an accepting set $S \subseteq \bbZ_6$ (where $S$ can vary for different gates). 

More generally, a $\CC_h[m]$-circuit is a depth-$h$ Boolean circuit consisting only of $\MOD_m$-gates. 
The question is: For fixed positive integers $h$ and $m$, what is the asymptotic size of the smallest possible $\CC_h[m]$-circuit that computes $\AND_n$? Is it polynomial in $n$?
In common complexity-theoretic terms, this question is phrased as ``$\CC^0 = \ACC^0$?''.
The class $\CC^0$ consists of all constant depth circuits that only use modular counting gates, while in $\ACC^0$, the circuits may additionally contain Boolean disjunction, conjunction and negation gates.
It is stunning that very little progress has been made despite the fact that this problem was first raised almost 30 years ago \cite{BarringtonST90}. 

As of today, only slightly superlinear lower bounds are known, and only for the number of wires, not gates \cite{chat-lowerbounds}, both for general $\CC_h[m]$-circuits as well as for the most restricted open case $\CC_2[6]$.
The lack of strong lower bounds is even more surprising when one compares it with the dual question: Can modulo counting be performed  efficiently by constant depth circuits using only the Boolean operations $\land, \lor, \neg$? 
This question was famously answered in the negative by H{\aa}stad \cite{hastadphd} already in the 1980s. Indeed, his switching lemma yields an asymptotically optimal exponential lower bound against constant depth Boolean circuits computing the parity function. 
More generally, it has been shown that $\MOD_q$ is not in $\AC^0[p]$ (the extension of $\AC^0$ with $\MOD_p$ gates), whenever $p \neq q$ are distinct primes \cite{Smo87}.
To sum up: We have known for a long time that Boolean operations cannot simulate modulo counting, but it is notoriously hard to settle whether modulo counting can simulate Boolean operations.

Common belief is in favour of a negative answer: Barrington, Straubing and Thérien first conjectured an exponential lower bound in \cite{BarringtonST90}, and since the work by Barrington, Beigel and Rudich in 1994 \cite{BarringtonBR94}, a $2^{\Omega(n^\epsilon)}$ size lower bound for $\CC_h[m]$-circuits computing $\AND_n$ is considered likely (where $0 < \epsilon < 1$). Some refer to this conjecture as the \emph{Exponential Size Hypothesis} (ESH).
For the very restricted setting of two layers with two different prime moduli, that is, 
$\MOD_{q} \circ \MOD_{p}$-circuits where $p \neq q$, an even stronger lower bound has been established unconditionally: 
Such circuits require size $2^{\Omega(n)}$ to compute $\AND_n$ \cite{GrolmuszT00, Grolmusz01, straubing2006note}. 

In this paper we study a circuit restriction of a different nature: Since the function $\AND_n$ is symmetric under all permutations of its inputs, it admits a \emph{symmetric circuit} representation. Such symmetric constructions are typically natural and intuitive, and it is also reasonable to assume that they are not too far from optimal.
Formally, we say that a circuit $C$ is \emph{fully symmetric} (or $\Sym_n$-symmetric) if for every $\pi \in \Sym_n$, there exists an automorphism of $C$ that permutes its input gates $x_1, \dots, x_n$ according to $\pi$.
We completely determine the $\CC_h[m]$-circuit complexity of $\AND_n$, for any depth $h \geq 2$, and any modulus $m$ with at least two prime divisors \footnote{ The case where $m$ is a prime power can be ignored. It is known that then, $\CC_h[m]$-circuits cannot compute $\AND_n$ for arbitrarily large $n$, see \cite[Proposition 2.1]{IdziakKK22LICS}} , as far as fully symmetric circuits are concerned.

\begin{theorem}
	\label{thm:main1}
	Fix an integer $m \geq 6$ with at least $r \geq 2$ distinct prime divisors.
	For every family of $\Sym_n$-symmetric $\MOD_m$-circuits $(C_n)_{n \in \bbN}$ computing the Boolean function $\AND_n$, the circuit size is at least $|C_n| \geq 2^{\Omega(n^{1/r} \cdot \log n)}$. There exists a family of $\CC_2[m]$-circuits that achieves this bound.
\end{theorem}	
The new contribution is the lower bound; the upper bound was presented by Idziak, Kawa\l ek, Krzaczkowski in 2022 \cite[Proposition 3.1]{IdziakKK22LICS}, and independently by Chapman and Williams \cite{ChapmanW22}, based on an idea from \cite{BarringtonBR94}. For completeness, we review the upper bound in Section \ref{sec:upperbound}. 
The most surprising insight from \cref{thm:main1} is the discovery that the natural and elegant depth-$2$ construction from \cite[Proposition 3.1]{IdziakKK22LICS} is in fact optimal under the assumption of $\Sym_n$-symmetry. 
Thus, one can never gain savings in the asymptotic size by increasing circuit depth beyond $2$, except by breaking symmetries. This answers a question implicitly posed by Kawa\l ek and Weiß in \cite{kawalekweiss25}, which is detailed below. 

Interestingly, a more involved construction from \cite{IdziakKK22LICS} demonstrates that (partially) breaking symmetries does indeed allow to build smaller circuits, at the expense of greater depth: For every constant $h > 2$, there exist $\CC_h[m]$-circuits for $\AND_n$ whose size is strictly smaller than $2^{\Theta(n^{1/r} \cdot \log n)}$, and the savings in size increase with $h$. 

Examining this depth-$h$ construction \cite[Proposition 4.3]{IdziakKK22LICS}, one finds that it is not $\Sym_n$-symmetric (which it cannot be by \cref{thm:main1}), but respects a smaller group of symmetries, that we call \emph{nested block symmetry}. 
This notion of symmetry is naturally exhibited by recursive circuit constructions that follow a divide-and-conquer approach; another such example can be found in \cite{ChapmanW22}. 

The symmetry group is best described as the automorphism group of a tree whose leaves correspond to the input variables $x_1, \dots, x_n$ of the circuit. 
We formalise this idea by fixing an $h$-tuple $\boldsymbol{k} = (k_1(n), \dots, k_h(n))$ of functions in $n$ such that $\prod_{i \in [h]} k_i(n) = n$, which defines for each $n \in \bbN$ a tree $\Tt_n^{\boldsymbol{k}}$: This tree has $h$ levels and on the $i$-th level, every node has $k_i(n)$ many children. The leaves of the tree are identified with the variables $x_1, \dots, x_n$ (for this to be possible, we need that the product over the $k_i(n)$ is $n$). The automorphism group $\Aut(\Tt_n^{\boldsymbol{k}})$ is the group of all permutations of the nodes of $\Tt_n^{\boldsymbol{k}}$ that preserve the edges and non-edges of the tree. 
A circuit is $\Aut(\Tt_n^{\boldsymbol{k}})$-symmetric circuit if for every $\pi \in \Aut(\Tt_n^{\boldsymbol{k}})$, there is an automorphism of the circuit that permutes the inputs $x_1, \dots, x_n$ precisely as $\pi$ permutes the leaves of $\Tt_n^{\boldsymbol{k}}$ (see Section \ref{sec:preliminaries} for details).
Note that not for every permutation $\pi \in \Sym_n$, there is a tree automorphism in $\Aut(\Tt_n^{\boldsymbol{k}})$ that acts like $\pi$ on the leaves. Hence, $\Aut(\Tt_n^{\boldsymbol{k}})$-symmetry is a laxer requirement of circuits than $\Sym_n$-symmetry.

Our next result extends \cref{thm:main1} to the nested block symmetric setting and pins down 
the asymptotic circuit size exactly, depending only on the choice of symmetry group $\Aut(\Tt_n^{\boldsymbol{k}})$. 
\cref{thm:main1} is in fact a special case of \cref{thm:main2}, when $\Tt_n^{\boldsymbol k}$ is taken to be the depth-$1$ tree with $n$ leaves. Nevertheless, \cref{thm:main1} is important enough to be stated separately, and its proof is more instructive.
\begin{theorem}
	\label{thm:main2}
	Fix an integer $m \geq 6$ with at least $r \geq 2$ distinct prime divisors.
	Moreover, fix a constant depth $h \in \bbN$, and a tuple $\boldsymbol{k} = (k_1(n), \dots, k_h(n))$ of block sizes such that $\prod_i k_i(n) =n$ for all $n \in \bbN$. Assume that $k_i(n) > 8$ for each $i \in [h]$ and all large enough $n \in \bbN$.
	Let $k_{\text{max}}(n) \coloneqq \max_{i \in [h]} k_i(n)$.\\
	Every $\Aut(\Tt_n^{\boldsymbol k})$-symmetric $\MOD_m$-circuit that computes the Boolean function $\AND_n$ has size at least 
	\[
	2^{\Omega(k_{\text{max}}(n)^{1/r} \cdot \log(k_{\text{max}}(n)))},
	\]
	and there exists a $\CC_{2h}[m]$-circuit that achieves this bound.
\end{theorem}	
The upper bound is achieved by applying the construction from \cref{thm:main1} in a recursive fashion, which requires $2$ circuit layers for each nested block of the symmetry group. An immediate consequence of the theorem is that one cannot gain savings in size by varying the block sizes on different levels of the symmetry group. Thus, the symmetric circuit complexity is just controlled by the choice of $h$ in this setting:

\begin{corollary}
	\label{cor:optimalNestedSize}
	For $h \in \bbN$ fixed, the choice of $\boldsymbol{k}$ which optimizes the size of an $\Aut(\Tt_n^{\boldsymbol k})$-symmetric $\MOD_m$-circuit for $\AND_n$ is $k_1(n) = k_2(n) = \dots = k_h(n) = n^{1/h}$. For this $\boldsymbol{k}$, there exists such a depth-$2h$ circuit of size $2^{\Theta(n^{1/(h \cdot r)} \cdot \log n)}$.
\end{corollary}	
\begin{proof}
	Since $\prod_{i \in [h]} k_i(n) = n$ for all $n \in \bbN$, the smallest value that $k_{\text{max}}(n)$ can attain is $n^{1/h}$, in case that all $k_i(n)$ are equal.
\end{proof}	
Strikingly, the $\Aut(\Tt_n^{\boldsymbol{k}})$-symmetric construction from \cite[Proposition 4.3]{IdziakKK22LICS} is slightly larger than this. It achieves only size (approximately) $2^{\Theta(n^{1/(h \cdot (r-1))} \cdot \log n)}$ instead of $2^{\Theta(n^{1/(h \cdot r)} \cdot \log n)}$ as in \cref{cor:optimalNestedSize}. 
This is because \cite[Proposition 4.3]{IdziakKK22LICS} aims at an optimized depth: Indeed, the authors manage to compress the circuits down to only $h+1$ layers, which is essentially just one layer per nesting depth of the symmetry group and likely optimal under $\Aut(\Tt_n^{\boldsymbol{k}})$-symmetry. In view of our \cref{cor:optimalNestedSize}, we consider it an important open problem to improve the size of this depth-$(h+1)$ construction so that it matches our lower bound, or to show that this is impossible. 

Our results confirm the $2^{\Omega(n^\epsilon)}$ lower bound conjectured by ESH for fully symmetric and nested block symmetric circuits. Since much of the literature focuses on the $\AND_n$-function, we have chosen to do the same, but our results can be proved for every symmetric function that is \emph{aperiodic} (for a definition, see Section \ref{sec:periodicity}), such as $\OR_n$ or $\MAJ_n$.\\ 

Our work raises the following immediate \textbf{open questions:}\\
\vspace*{-1em}
\begin{enumerate}
	\item Is there a \emph{size-depth tradeoff} for nested block symmetric circuits?\\ The current knowledge suggests this: We have size-optimal constructions that are at least a factor of $2$ away from the optimal depth (\cref{cor:optimalNestedSize}), and we have (likely) depth-optimal circuits of slightly suboptimal size \cite[Proposition 4.3]{IdziakKK22LICS}.
	\item What is the smallest possible symmetry group for which our method works? Do lower bounds for symmetric circuits inform the search for general (non-symmetric) lower bounds? For example, can every non-symmetric $\CC^0$-circuit for $\AND_n$ be symmetrized at only a small cost?
\end{enumerate}

\paragraph*{Related work on modular circuits }
The complexity of low-depth $\CC^0$-circuits has particular relevance because of its various connections to other questions, of which we can only mention a few here.
For example, it is known that strong lower bounds against low-depth $\CC^0$-circuits would imply faster algorithms for solving equations over solvable groups \cite{IdziakKKW22-icalp}, for circuit satisfiability problems over algebras \cite{IdziakKK20, IdziakKK25}, and for constraint satisfaction problems with global modular constraints \cite{brakensiek2022}.
Another surprising application is in coding theory. Techniques from the construction of small $\CC^0$-circuits for Boolean functions have been used to obtain explicit Ramsey-style graphs \cite{Gopalan14, grolmusz2000}. These are crucial for example in the design of particularly good locally decodable error-correcting codes \cite{Efremenko12, zev11} and private information retrieval schemes \cite{zeev15}.

\paragraph*{Related work on symmetric circuits}
The idea to study symmetric circuits to facilitate lower bounds has been employed successfully in many different contexts in recent years.
To name only a few examples, there are symmetry-based lower bounds for constant depth formulas with Boolean and majority gates \cite{HeR23}, lower bounds for uniform symmetric Boolean (threshold) circuit families via a connection to fixed-point logics from finite model theory \cite{anderson_symmetric_2017}, and a recent research strand on symmetry in algebraic complexity theory \cite{DawarW20, dawar2021lower, SymHomPolynomials, STOCpaper}, notably proving the permanent polynomials to be exponentially hard for symmetric circuits.

For our purposes, the most relevant related work is a very recent paper by Kawa\l ek and Weiß \cite{kawalekweiss25}. They seem to be the first to consider symmetry in the context of $\CC^0$, but with a scope limited to $\MOD_q \circ \MOD_p \circ \AND_d$-circuits (where $d$ is some fixed integer), and $\Sym_n$ as the symmetry group. 
Our lower bounds cover that case (when $m$ has $p$ and $q$ as prime divisors) and apply to a much more general circuit model. In particular, Kawa\l ek and Weiß suggest that to obtain smaller symmetric circuits for $\AND_n$, one may have to increase the depth. Our \cref{thm:main1} surprisingly refutes this, and shows that a depth greater than $2$ yields no size improvement, unless also the symmetry constraint is relaxed as in \cref{thm:main2}. 
After publication of the conference version of this paper, Krishan and Sarma published an alternative proof of the symmetric $\CC^0$ lower bound on arxiv \cite{krishan2026cc0lowerboundstorus}. They consider the same symmetry groups and obtain similar lower bounds, though without the explicit dependence on the number of prime divisors of the modulus. Their approach is very different: It relies on the approximation of symmetric $\CC^0$-circuits by torus polynomials and a known lower bound on the degree of torus polynomial expressing AND. It is worth investigating whether their method can be combined with ours to possibly generalise it to smaller symmetry groups.
 
 \paragraph*{Our techniques}
The key challenge we solve is to overcome the obstacles that prevent the technique in
 \cite{kawalekweiss25} from generalising to symmetric circuits of arbitrary depth and with $\MOD_m$-gates for composite numbers $m$. We accomplish this by using a very different technical framework, based on the group-theoretic notion of \emph{supports}. This is the central tool in the aforementioned articles by Dawar and others, and we use it in an inductive approach vaguely similar to the one in \cite{SymHomPolynomials}.
 The lower bound for $\AND_n$ is then based on a periodicity argument as in \cite{kawalekweiss25}:
 Our main technical result, \cref{lem:inductionPeriodLength}, shows that symmetric $\MOD_m$-circuits of small (support) size necessarily compute periodic functions -- but $\AND_n$ is aperiodic.
 
 Another contribution of this work is the extension of the group-theoretic toolkit for dealing with the smaller symmetry groups $\Aut(\Tt_n^{\boldsymbol{k}})$. Thus far, the literature has mostly focused on direct products of symmetric or alternating groups, which are technically easier to handle. 
 
\paragraph*{Acknowledgements}
I am grateful to Piotr Kawa\l ek for posing this question to me, and for his invaluable help in learning and presenting the research context. 
The author was funded by UK Research and Innovation (UKRI) under the UK government’s Horizon Europe funding guarantee: grant number EP/X028259/1.
No AI was used in the production of this paper.

\section{Preliminaries}
\label{sec:preliminaries}
We write $[n] = \{1, \dots, n\}$.
For a tuple $\bar{x} = (x_1, \dots, x_n)$ indexed by $[n]$, and a subset $I \subseteq [n]$ of the indices, we use the notation $\bar{x}_I$ to refer to the subtuple $(x_i)_{i \in I}$ of $\bar{x}$ consisting of the entries with indices in $I$. 
In this notation, we also sometimes merge subtuples:
For $I, J \subseteq [n]$, we denote by $\bar{x}_I\bar{x}_J$ the subtuple $\bar{x}_{I \cup J}$ of $\bar{x}$. In all these cases, the ordering of the respective subtuple of $\bar{x}$ is inherited from $\bar{x}$.

The \emph{depth} of a rooted tree or DAG is the maximum number of edges on any path from the root to a leaf.

\subsection{Permutation groups and supports}
\label{sec:permGroups}
We write $\Sym_n$ for the symmetric group acting on the set $[n]$, and if $A$ is a set, then $\Sym(A)$ denotes the symmetric group acting on $A$.
For $S \subseteq [n]$, we write $\Stab(S) \leq \Sym_n$ for the setwise stabiliser subgroup of $S$ in $\Sym_n$, and $\StabP(S) \leq \Sym_n$ for the pointwise stabiliser of $S$. The setwise stabiliser is the subgroup of permutations $\pi$ such that $\pi(S) = S$, whereas the pointwise stabiliser is the subgroup consisting of all $\pi \in \Sym_n$ such that $\pi(s) = s$ for every $s \in S$.
This notion can be restricted to subgroups $\Gamma \leq \Sym_n$ as well: $\Stab_\Gamma(S) \leq \Gamma$ and $\StabP_\Gamma(S) \leq \Gamma$ denote the respective subgroups of $\Gamma$ that fix $S$ setwise or pointwise.

For a group $\Gamma \leq \Sym_n$, and an element $a \in [n]$, the \emph{$\Gamma$-orbit} of $a$ is the set $\Orb_\Gamma(a) \coloneqq \{ \pi(a) \mid \pi \in \Gamma  \}$ of all possible images of $a$.
By the well-known Orbit-Stabiliser Theorem, $|\Orb_\Gamma(a)| = \frac{|\Gamma|}{|\Stab_\Gamma(a)|}$.
The notion of an orbit also applies to subsets of $[n]$, or more generally, other objects, such as gates of a circuit, that $\Gamma$ may be acting on.  

A set $S \subseteq [n]$ is a \emph{support} of a group $\Gamma \leq \Sym_n$ if $\StabP(S) \leq \Gamma$. It is known that every $\Gamma \leq \Sym_n$ that has a support of size $< n/2$ has a \emph{unique minimal} support, denoted $\supp(\Gamma)$ \cite[Lemma 26]{blass1999choiceless}.
This notion will be of central importance in our analysis of symmetric circuits:
In symmetric circuits of bounded size, also the minimal supports of the stabiliser groups of the gates will have bounded size. The function computed by a gate can then be described in terms of this small support.

\paragraph*{Nested symmetric groups}
Besides $\Sym_n$, we deal with nested symmetric groups. These are best described as the automorphism groups of rooted trees. 
Fix a depth $h \in \bbN$ and let $\boldsymbol{k} = (k_1(n), \dots, k_h(n))$, where for each $i \in [h]$, $k_i \colon \bbN \to \bbN$ is a function such that $\prod_{i \in [h]} k_i(n) = n$ for every $n \in \bbN$. The tuple $\boldsymbol{k}$ defines a family of rooted symmetric trees, one for each $n$: The $n$-th tree $\Tt_n^{\boldsymbol{k}}$ has $h$ levels, and $k_i(n)$ is the number of children of every node on level $i$. 
The $n$ leaf nodes are on level $0$, and the root of the tree is the only node on level $h$. The automorphism group of $\Tt_n^{\boldsymbol{k}}$ is denoted $\Aut(\Tt_n^{\boldsymbol{k}})$. It acts on the vertex set $V(\Tt_n^{\boldsymbol{k}})$. Each $\pi \in \Aut(\Tt_n^{\boldsymbol{k}})$ maps every subtree of $\Tt_n^{\boldsymbol{k}}$ to a subtree rooted at the same level, possibly permuting subtrees further down. Because in each level, all nodes have the same number of children, every node can be mapped to every other node on the same level by an automorphism in $\Aut(\Tt_n^{\boldsymbol{k}})$.
Formally, $\Aut(\Tt_n^{\boldsymbol{k}})$ is isomorphic to an iterated wreath product of symmetric groups, defined inductively as follows.
Let $T$ be a subtree of $\Tt_n^{\boldsymbol{k}}$ of depth $1$, having $k_1(n)$ leaves. Then $\Aut(T) \cong \Sym_{k_1(n)}$. Now assume $T$ is a subtree of $\Tt_n^{\boldsymbol{k}}$ of depth $i > 1$. Then the root $v$ of $T$ has $k_i(n)$ children, and $\Aut(T) \cong \Aut(T') \wr \Sym_{k_i(n)}$, where $T'$ is isomorphic to the subtree rooted at any/every child of $v$.

The group $\Aut(\Tt_n^{\boldsymbol{k}})$ embeds into $\Sym_n$ by identifying every $\pi \in \Aut(\Tt_n^{\boldsymbol{k}})$ with the permutation that it induces on the leaves of the tree. Note that every $\sigma \in \Sym_n$ is induced by at most one $\pi \in \Aut(\Tt_n^{\boldsymbol{k}})$. 

For the analysis, it will be helpful to speak of the action of $\Aut(\Tt_n^{\boldsymbol{k}})$ on \emph{blocks}: For a node $v \in V(\Tt_n^{\boldsymbol{k}}), B(v) \subseteq V(\Tt_n^{\boldsymbol{k}})$ denotes the block of $v$, by which we mean the set of all nodes (including $v$) that have the same parent as $v$. Every $\pi \in \Aut(\Tt_n^{\boldsymbol{k}})$ stabilises a block setwise or moves it to another block on the same level. The set of all blocks is denoted
\[
\Bb(\Tt_n^{\boldsymbol{k}}) \coloneqq \{ B(v) \mid v \in V(\Tt_n^{\boldsymbol{k}})  \}.
\]

For $i \in [h] \cup \{0\}$, we denote by $L_i \subseteq V(\Tt^{\boldsymbol{k}}_n)$ the nodes in the $i$-th level of the tree. For a set of nodes $W \subseteq V(\Tt^{\boldsymbol{k}}_n)$, we denote by $L_0(W) \subseteq L_0$ the set of all leaves $w$ that have an ancestor (i.e.\ node on a path from the root to $w$) in $W$.

\paragraph*{Supports for nested symmetric groups}
Any group $\Gamma \leq \Aut(\Tt^{\boldsymbol{k}}_n)$ can be embedded into $\Sym_n$ (via its action on the leaves) and thus admits a notion of support in the sense described previously. 
However, when $\Gamma$ is explicitly presented as a subgroup of $\Aut(\Tt^{\boldsymbol{k}}_n)$, we use a more refined notion, that we call \emph{blockwise support}. It breaks up the support according to the different copies of symmetric groups in $\Aut(\Tt^{\boldsymbol{k}}_n)$. 
For a block $B \in \Bb(\Tt_n^{\boldsymbol{k}})$, let $\Gamma|_B \leq \Sym(B)$ denote the permutation group on $B$ consisting of all $\pi \in \Sym(B)$ such that there exists a $\sigma \in \Gamma$ that fixes $B$ setwise and satisfies $\sigma|_B = \pi$, i.e.\ $\sigma$ permutes the nodes in $B$ like $\pi$ does. 

A set $S \subseteq B$ is a \emph{$B$-support} of $\Gamma \leq \Aut(\Tt_n^{\boldsymbol{k}})$ if $\StabP_{\Sym(B)}(S) \leq \Gamma|_B$. In other words, this means that for every $\pi \in \StabP_{\Sym(B)}(S)$, there exists at least one permutation $\sigma \in \Gamma$ that acts like $\pi$ on $B$.

\subsection{(Symmetric) modular circuits}

A circuit is a DAG, possibly with multiedges, and a single designated root. Its nodes are called gates and its edges are also called wires. Edges are directed from a gate where a value is computed towards the next gate that uses this value as an input. The nodes with no incoming edges are called input gates, and each input gate $g$ is labelled with a variable $\ell(g) \in \{x_1, \dots, x_n\}$, where $n$ is the arity of the function to be computed by the circuit. Each internal gate is labelled with an operation. In this paper, we only consider circuits with modular counting gates: For $m \in \bbN$, and $R \subseteq \bbZ_m$, the operation $\MOD_m^R$ is of arbitrary fan-in $k$, and satisfies
\begin{align*}
	\MOD_m^R(x_1, \dots, x_k) = \begin{cases}
		1 & \text{ if } (\sum_{i \in [k]} x_i \mod m) \in R\\
		0 & \text{ otherwise}
	\end{cases}
\end{align*}
A \emph{$\MOD_m$-circuit} is one where every internal gate is labelled with the operation $\MOD_m^R(x_1, \dots, x_k)$ for an arbitrary $R \subseteq \bbZ_m$. The computation result is the Boolean value that is computed at the root. 

We assume throughout that for every variable $x_i$, there exists exactly one input gate with label $x_i$. This is not a restriction since distinct input gates labelled with the same variable can always be identified.

The \emph{size} of a circuit $C$ is $|C| \coloneqq |V(C)| + |E(C)|$, the total number of gates plus wires, counted with multiplicities.
For a gate $g \in V(C)$, we write $gE(C) \subseteq V(C)$ for the set of its children, which are the gates $h$ whose value is fed into the gate $g$, i.e.\ such that $(h,g) \in E(C)$. For a gate $g$, we also write $g(\bar{x})$ to denote the function from $\{x_1, \dots, x_n\}$ to $\{0,1\}$ that is computed by the subcircuit of $C$ rooted at $g$.

\paragraph*{Symmetric circuits}
Let $n \in \bbN$, and $\Gamma \leq \Sym_n$ be a subgroup of $\Sym_n$.
A $\MOD_m$-circuit $C$ is called \emph{$\Gamma$-symmetric} if its set of input variables is $\{x_1, \dots, x_n\}$ and every $\pi \in \Gamma$ acting on the input gates \emph{extends to an automorphism} of $C$. That is, for every $\pi \in \Gamma$, there exists a $\sigma \in \Sym(V(C))$ such that $\pi(\ell(g)) = \ell(\sigma(g))$ for all inputs gates $g \in V(C)$, and $\sigma$ is an automorphism of $C$. This means that for every internal gate $g$, $\ell(g) = \ell(\sigma(g))$ (i.e., the gates compute the same operation $\MOD_m^R$, for the same $R$), and for any two gates $g,h \in V(C)$, the multiplicity of the directed edge $(g,h)$ is the same as of $(\sigma(g), \sigma(h))$.
We call $C$ \emph{rigid} if for every $\pi \in \Gamma$, the circuit automorphism $\sigma$ extending $\pi$ is unique. This is equivalent to $C$ not having any non-trivial automorphism that fixes every input gate.
Fortunately, w.l.o.g.\ we can always assume our symmetric circuits to be rigid (see e.g.\ \cite[Lemma 4.3]{arxivVersionSymAlgebraicCircuits} or \cref{lem:rigidification} in the appendix).

The advantage of working with rigid $\Gamma$-symmetric circuits is that for every $\pi \in \Gamma$, and $g \in V(C)$, we may write $\pi(g)$ to mean the well-defined gate $\sigma(g)$, for the unique circuit automorphism $\sigma$ that extends $\pi$. In this sense, if $C$ is rigid, then $\Gamma$ has a well-defined (faithful) action on $V(C)$. 

With respect to this action, we can speak about the orbits of gates. 
Their size is an important complexity measure of $\Gamma$-symmetric rigid circuits, called \emph{orbit size}, by which we mean the maximum size of a $\Gamma$-orbit of any gate:
\[
\maxorb_\Gamma(C) \coloneqq \max_{g \in V(C)} | \Orb_\Gamma(g) |.
\]
Note that for any gate $g$, $\Orb_\Gamma(g) \subseteq V(C)$, so to establish lower bounds on the circuit size $|V(C)|$, it is sufficient to prove lower bounds for $\maxorb_\Gamma(C)$.

By induction on the circuit structure, it is not difficult to verify the following fact about the interplay between symmetry and the semantics of the gates (see Appendix \ref{sec:appendixGroups}).
\begin{restatable}{lemma}{symmetricCircuitSemantics}
	\label{lem:symmetricCircuitSemantics}
	Let $C$ be a $\Gamma$-symmetric rigid circuit, for $\Gamma \leq \Sym_n$. Let $\delta \colon \{x_1, \dots, x_n\} \to \{0,1\}$ be an assignment to the variables, let $g \in V(C)$ be a gate, and let $\pi \in \Gamma$. Then 
	\[
	g(\delta(x_1), \dots, \delta(x_n)) = \pi(g)(\delta(\pi^{-1}(x_1)), \dots, \delta(\pi^{-1}(x_n))).
	\]
\end{restatable}

\paragraph*{Supports in symmetric circuits}
We apply the previously introduced notions of supports to stabiliser groups of gates in symmetric circuits. 
We have already seen that the concept of a support can be defined differently for different permutation groups. Therefore, from here on, we restrict the symmetry group $\Gamma$ of our circuits and always assume that $\Gamma \in \{\Sym_n, \Aut(\Tt_n^{\boldsymbol{k}}) \}$, for some $n$ or some symmetric $n$-leaf tree $\Tt_n^{\boldsymbol{k}}$.

We are interested in the \emph{support of a gate} $g$ in a given $\Gamma$-symmetric rigid circuit $C$. Let $\Stab(g) \leq \Gamma$ be the stabiliser group of the gate, that is, $\Stab(g) \coloneqq \{ \pi \in \Gamma \mid \pi(g) = g\}$.
\begin{definition}[Supports of gates]
	\label{def:support}
	Let $g \in V(C)$ be a gate in a $\Gamma$-symmetric rigid circuit, for $\Gamma \in \{\Sym_n, \Aut(\Tt_n^{\boldsymbol{k}}) \}$.
	\begin{itemize}
		\item If $\Gamma = \Sym_n$, then the support of $g$ is $\supp(g) \coloneqq \supp(\Stab_{\Sym_n}(g)) \subseteq [n]$, i.e., the unique minimal support of $\Stab_{\Sym_n}(g)$ in $\Sym_n$. 
		\item If $\Gamma = \Aut(\Tt_n^{\boldsymbol{k}})$, then we consider the \emph{blockwise support} of $g$: For every $B \in \Bb(\Tt_n^{\boldsymbol{k}})$, we denote by $\supp_B(g) \subseteq B$ the unique minimal $B$-support of the group $\Stab_{\Aut(\Tt_n^{\boldsymbol{k}})}(g)$.
	\end{itemize}	
	In either case, the support of $g$ is undefined if the respective minimal ($B$-)support of $\Stab(g)$ is not uniquely defined.
\end{definition}	
We still need to argue that in the scenarios we study in this paper, the respective minimal supports exist and are unique, so that $\supp(g)$ and $\supp_B(g)$ are indeed always defined when we need them.
\begin{restatable}{lemma}{supportsizes}
\label{lem:supportSizes}	
\begin{enumerate}
\item Let $n > 8$ and $C$ be a $\Sym_n$-symmetric rigid circuit. Let $k$ be such that $1 \leq k \leq \frac{n}{4}$ and $\maxorb_{\Sym_n}(C) \leq \binom{n}{k}$.
Then for every $g \in V(C)$, $\Stab_{\Sym_n}(g)$ has a support of size less than $k$. 
\item Let $C$ be an $\Aut(\Tt_n^{\boldsymbol{k}})$-symmetric rigid circuit, for some $h$-tuple $\boldsymbol{k}$ such that $k_{\text{min}} \coloneqq \min_{i \in [h]} k_i(n)  > 8$.
For every block $B \in \Bb(\Tt_n^{\boldsymbol{k}})$, let $k_B$ be such that $1 \leq k_B \leq \frac{|B|}{4}$ and $\maxorb_{\Stab(B)}(C) \leq \binom{|B|}{k_B}$. Then for every $g \in V(C)$, $\Stab_{\Aut(\Tt_n^{\boldsymbol{k}})}(g)$ has a $B$-support $S \subseteq B$ with $|S| < k_B$.
\end{enumerate}
\end{restatable}
The first part of the lemma is shown in \cite[Theorem 14]{dawar_symmetric_2024} for symmetric Boolean circuits, and the set-up here is very similar. The second item follows from the first with an additional argument. The proof details are given in Appendix \ref{sec:appendixGroups}.

As already mentioned, by \cite[Lemma 26]{blass1999choiceless}, every subgroup of $\Sym_n$ that has a support of size $< n/2$ also has a unique minimal support.
This holds in particular whenever the conditions of the above lemma are satisfied, both in case (1) and (2). In the set-up in the following sections, this will always be the case because if the conditions of the lemma are not satisfied for a circuit $C$, then $\maxorb(C)$, and hence $|V(C)|$, is anyway at least as large as the circuit size lower bounds we have to prove for \cref{thm:main1} and \cref{thm:main2}.
For a circuit $C$ where the respective supports are defined, we write
\begin{align*}
	\maxSup(C) &\coloneqq \max_{g \in V(C)} |\supp(g)|, \text{ if } C \text{ is } \Sym_n\text{-symmetric} .\\
	\maxSup_{B}(C) &\coloneqq \max_{g \in V(C)} |\supp_{B}(g)|, \text{ if } C \text{ is } \Aut(\Tt_n^{\boldsymbol{k}})\text{-symmetric}
	 \text{ and } B \in \Bb(\Tt_n^{\boldsymbol{k}}).
\end{align*}

Another known fact about supports of gates in symmetric circuits is that they are moved by permutations in the expected way:
\begin{lemma}[{\cite[Lemma 4.2]{arxivVersionSymAlgebraicCircuits}}]
	\label{lem:supportsMoveWithPermutation}
	Let $C$ be a $\Sym_n$-symmetric rigid circuit in which the supports of all gates are defined. Let $g \in V(C)$ be a gate. Then for every $\pi \in \Sym_n$, $\supp(\pi(g)) = \pi(\supp(g))$. 
\end{lemma}	 
Analogously, in every $\Aut(\Tt_n^{\boldsymbol{k}})$-symmetric circuit $C$, $\supp_{\pi(B)}(\pi(g)) = \pi(\supp_B(g))$ for every gate $g$, every block $B \in \Bb(\Tt_n^{\boldsymbol{k}})$, and every $\pi \in \Aut(\Tt_n^{\boldsymbol{k}})$.

\subsection{Periodic functions}
\label{sec:periodicity}
A function $f \colon \bbN \to \bbN$ is called \emph{periodic} with a period of length $\ell$ if $f(x) = f(x+ \ell)$ for every $x \in \bbN$. An analogous definition applies if $f$ is only defined on an initial segment of $\bbN$. 
A particular periodic function that will be of high importance for us is the binomial coefficient modulo $m$.
Its period length is known exactly:
\begin{theorem}[{\cite[Theorem 2.3]{laugier2015periodicsequencesmodulom}}]
	\label{thm:binomialPeriod}
	Let $m \in \bbN$ be fixed, and let $p_1, \dots, p_r$ be its prime divisors. Fix $x \in \bbN$. The function $a(n) \coloneqq \binom{n}{x} \mod m$ has a period of minimal length 
	\[
	\ell(m,x) = m \cdot \prod_{i \in [r]} p_i^{\lfloor \log_{p_i}(x) \rfloor}.
	\]
\end{theorem}

Let $f \colon \{0,1\}^n \to \{0,1\}$ be an $n$-ary Boolean function. If for every $\pi \in \Sym_n$, and every $\bar{x} \in \{0,1\}^n$, $f(x_1, \dots, x_n) = f(\pi(\bar{x})) \coloneqq f(x_{\pi^{-1}(1)}, \dots, x_{\pi^{-1}(n)})$, then we say that $f$ is \emph{$\Sym_n$-symmetric}. The value $f(\bar{x})$ of a  $\Sym_n$-symmetric $n$-ary function $f$ only depends on the number of $1$-entries in $\bar{x}$, denoted $|\bar{x}|_1$. Thus, we may view $f(\bar{x})$ as a unary function $f(|\bar{x}|_1)$ defined on $\{0, \dots, n\}$, and as such, we can speak of its period. 

\begin{lemma}
	\label{lem:ANDnotPeriodic}
	The $\Sym_n$-symmetric Boolean function $\AND_n$ does not have a period of any length $\leq n$.
\end{lemma}	 
\begin{proof}
	For a contradiction, assume that $0 < \ell \leq n$ was the period length of $\AND_n$. Then $\AND_n(\bar{x}) = 1$ also in case that $|\bar{x}|_1 = n-\ell$. But $0 \leq n-\ell < n$, so $\AND_n(\bar{x}) = 0$ in this case. Contradiction.
\end{proof}

\section{Size lower bound for fully symmetric circuits}
In this section, we prove the size lower bound from \cref{thm:main1} for $\Sym_n$-symmetric $\MOD_m$-circuits computing $\AND_n$. The technical core, from which the size lower bound follows, is a lower bound on the \emph{support size} $\maxSup(C)$ required to compute $\AND_n$:

\begin{theorem}
\label{thm:mainSupportLowerBound}
Fix a positive integer $m > 3$ and let $r$ be the number of distinct prime divisors of $m$.
Let $(C_n)_{n \in \bbN}$ be a family of $\Sym_n$-symmetric rigid $\MOD_m$-circuits. 
If $\maxSup(C_n) < (n/m)^{1/r}$ for all $n \in \bbN$, then $C_n$ does not compute $\AND_n$.
\end{theorem}

\begin{corollary}
\label{cor:sizeLowerBound}
In the setting of the theorem,
if $C_n$ computes $\AND_n$ for an $n > 8$, then
$|V(C_n)| \geq \binom{n}{(n/m)^{1/r}}$.
\end{corollary}
\begin{proof}
	If $C_n$ computes $\AND_n$, then by \cref{thm:mainSupportLowerBound}, $\maxSup(C_n) \geq (n/m)^{1/r}$. 
	Suppose for a contradiction that $|V(C_n)| < \binom{n}{(n/m)^{1/r}}$. Then in particular, $\maxorb(C_n) < \binom{n}{(n/m)^{1/r}}$. Because $n > 8$ and $m > 3$, the conditions of \cref{lem:supportSizes} (1) are fulfilled for $k = (n/m)^{1/r}$, so $\maxSup(C_n) < (n/m)^{1/r}$, which is a contradiction.
\end{proof}	

By replacing factorials with their Stirling approximations, we can compute that 
$\binom{n}{(n/m)^{1/r}} \geq (f_1(m,r) \cdot n^{f_2(m,r)})^{n^{1/r}}$, where $f_1, f_2$ are functions depending on $m,r$, so we can treat them as constants. Thus, \cref{cor:sizeLowerBound} indeed yields the asymptotic circuit size lower bound of $2^{\Omega(n^{1/r} \cdot \log n)}$ that is claimed in \cref{thm:main1}.

It remains to prove \cref{thm:mainSupportLowerBound}. This is done by showing that if $\maxSup(C_n) < (n/m)^{1/r}$, then the function computed by $C_n$ is periodic with a period of length $<n$.

The proof rests on the upper bound on the period length shown in \cref{cor:upperboundPeriodLength}. 
Before we can state and prove that corollary, we need to introduce a refined notion of period, for not fully symmetric functions: 

Let $S \subseteq [n]$, and let $f(x_1, \dots, x_n)$ be a $\StabP(S)$-symmetric function. Let $X_S \coloneqq \{x_i \mid i \in S\}$.
Then the value of $f$ is fully determined by the assignment $\alpha \colon X_S \to \{0,1\}$ and by the number of $1$-entries in the variables $\{x_i \mid i \in [n] \setminus S \}$.
Formally, let us write $f_\alpha \colon X_{[n] \setminus S} \to \{0,1\}$ for the function obtained from $f$ by fixing the variables in $X_S$ to the values given by $\alpha$. That is, \[
f_\alpha(\bar{x}_{[n] \setminus S}) \coloneqq f(\alpha(\bar{x}_S)\bar{x}_{[n] \setminus S}).
\]
When we say that a $\StabP(S)$-symmetric function $f$ has a period, we mean that for every $\alpha \colon X_S \to \{0,1\}$, the function $f_\alpha(\bar{x}_{[n] \setminus S})$, whose value only depends on the number of $1$-entries in its input, has a period. This period may be of different length for different assignments $\alpha \colon X_S \to \{0,1\}$, but when we say that $f$ has a period of length at most $\ell$, then we mean that for each $\alpha \colon X_S \to \{0,1\}$, the period length of $f_\alpha$ is at most $\ell$.

Now let $g \in V(C)$ be a gate. The function $g(\bar{x})$ computed by $g$ is always $\Stab(g)$-symmetric by \cref{lem:symmetricCircuitSemantics}. Because $\StabP(\supp(g)) \leq \Stab(g)$, it is also $\StabP(\supp(g))$-symmetric. So when we say that $g(\bar{x})$ has a period of length at most $\ell$, we mean that $g_\alpha(\bar{x})$ has such a period for every $\alpha \colon X_{\supp(g)} \to \{0,1\}$.
Now the technical lemma that we want to show reads as follows. 
\begin{restatable}{lemma}{mainInductionPeriodLength}
\label{lem:inductionPeriodLength}
Let $s \colon \bbN \to \bbN$ be a function in $o(n)$. Fix a number $m \in \bbN$.
Let $(C_n)_{n \in \bbN}$ be a family of $\Sym_n$-symmetric rigid $\MOD_m$-circuits in which all supports have size at most $s(n)$. Fix $n \in \bbN$.
Let $g$ be a gate in $C_n$. Then the $\StabP(\supp(g))$-symmetric function $g(\bar{x})$ has a period of length at most 
\[
q(m,s) \coloneqq m \cdot \prod_{i \in [r]} p_i^{\lfloor \log_{p_i}(s(n)) \rfloor},
\]
where the product ranges over the prime factors $p_1, \dots, p_r$ of $m$.
\end{restatable}
Note that this bound on the period length does not depend on the depth of the gate -- this also explains why it is possible already for depth-$2$ circuits to achieve the optimal size. All that matters is the size of the supports.

\begin{corollary}
\label{cor:upperboundPeriodLength}
In the setting of Lemma \ref{lem:inductionPeriodLength}, the $\Sym_n$-symmetric function computed at the output gate of $C_n$ has a period of length at most $m \cdot s(n)^r$, where $r$ denotes the number of distinct prime divisors of $m$.
\end{corollary}

\begin{proof}[Proof of \cref{thm:mainSupportLowerBound}]
	If $\maxSup(C_n) < (n/m)^{1/r}$, then by \cref{cor:upperboundPeriodLength}, the $\Sym_n$-symmetric function computed by $C_n$ has a period of length at most $m \cdot ( (n/m)^{1/r})^r = n$. But then, this function cannot be $\AND_n$ by \cref{lem:ANDnotPeriodic}.
\end{proof}	

The remainder of the section is dedicated to the proof of \cref{lem:inductionPeriodLength}.
The lemma is proved by induction on the layer $d$ of the gate $g$ in the circuit $C \coloneqq C_n$. The induction hypothesis is slightly stronger than the claim of the lemma: We show that for every gate $g$ and assignment $\alpha$ to its support, the period length of $g_\alpha(\bar{x}_{[n] \setminus \supp(g)})$ is either $1$ or of the form 
\[
m \cdot \prod_{i \in [r]} p_i^{c_i},
\]
where each exponent $c_i$ satisfies $c_i \leq \lfloor \log_{p_i}(s(n)) \rfloor$.

In the case $d=0$, the gate $g$ is an input gate labelled with a variable $x_i$. The function it computes is clearly $\Stab(i)$-symmetric and has a period of length $1$: Any fixed assignment $\alpha$ to the variable $x_i$ completely determines the value of the gate $g$, and so, $g_\alpha(i)=g_\alpha(i+1)$ for every $i \in \{0, \dots, n-1\}$ (recall that because of the symmetry, we can view $g_\alpha$ as a unary function in $|\bar{x}_{[n] \setminus \{i\}}|_1$).\\ 

Now the inductive step requires more work. Let $g$ be a gate on layer $d+1$ and assume the induction hypothesis holds for all gates in layers $0, \dots, d$. We partition the set $gE(C)$ of children of $g$ into $\StabP(\supp(g))$-orbits, and treat these orbits separately. 
To see that $gE(C)$ indeed decomposes into such orbits, note that $\StabP(\supp(g)) \leq \Stab(g)$. Therefore, every permutation $\pi \in \StabP(\supp(g))$ fixes the gate $g$ and hence must map every child of $g$ to a child of $g$ (since $\pi$ acts on $C$ as a circuit automorphism and thus preserves wires). Thus, indeed, $\StabP(\supp(g))$ acts as a permutation group on $gE(C)$, and this permutation domain decomposes into orbits.

We aim to analyse, for every fixed $X_{\alpha} \colon \supp(g) \to \{0,1\}$, the function $g_\alpha$ and its period, by considering the contribution of each orbit of children separately. 
So fix some $\StabP(\supp(g))$-orbit $O \subseteq gE(C)$. For fixed $\alpha \colon X_{\supp(g)} \to \{0,1\}$, let $O_\alpha(\bar{x}_{[n] \setminus \supp(g)})$ denote the contribution of $O$ to $g_\alpha$:
\[
O_\alpha(\bar{x}_{[n] \setminus \supp(g)}) \coloneqq \sum_{h \in O} h(\alpha(\bar{x}_{\supp(g)})\bar{x}_{[n] \setminus\supp(g)}).
\]
This is simply the sum (before taking modulo $m$) over the values computed by all children $h$ of $g$ in the orbit $O$. 

The induction hypothesis gives us the period length for $h_{\alpha'}(\bar{x}_{[n] \setminus \supp(h)})$ for assignments $\alpha'$ to $X_{\supp(h)}$. In general, the assignment $\alpha|_{\supp(h)}$ has as its domain only a subset of $X_{\supp(h)}$, namely $X_{\supp(h)} \cap X_{\supp(g)}$, so in order to use the induction hypothesis, we will also need to consider an assignment $\beta$ that covers the variables indexed with $S(h) \coloneqq \supp(h) \setminus \supp(g)$, for each $h \in O$.

Our goal is to derive a formula for $O_\alpha(\beta(\bar{x}_{[n] \setminus \sup(g)}))$, for any given assignment $\beta \colon X_{[n] \setminus \sup(g)} \to \{0,1\}$, which will allow us to apply the induction hypothesis in a straightforward way. Note that $\alpha$ and $\beta$ taken together define an assignment to all variables, which we denote as 
\[
\alpha\beta \colon \{x_1, \dots, x_n\} \to \{0,1\}.
\]

For each $h \in O$, the value $h(\alpha\beta(\bar{x}))$
depends solely on $\alpha\beta(\bar{x}_{\supp(h)}) = \alpha(\bar{x}_{\supp(h) \cap \supp(g)})\beta(\bar{x}_{S(h)})$ (which is the fixed assignment to the support of the gate $h$), and on the number of $1$s assigned to the variables outside of the support. This is because $h$ computes a $\StabP(\supp(h))$-symmetric function, as explained earlier. 

We will now group the gates $h \in O$ according to the value of the assignment $\alpha\beta(\bar{x}_{\supp(h)})$, and show that the gates that are grouped together compute the same value under $\alpha\beta$. Moreover, we will show that each such collection of gates is of the same size. With this, we will arrive at a useful expression for $O_\alpha(\beta(\bar{x}_{[n] \setminus \supp(g)}))$.\\

To speak about this formally, we have to define a way to view supports of gates, which are per se unordered sets, as ordered tuples. 
Fix an arbitrary $h^* \in O$ as the representative of the orbit.  Let $s \coloneqq |\supp(h^*)|$. Fix an arbitrary ordering of $\supp(h^*)$ such that the elements of $\supp(g) \cap \supp(h^*)$ come before $\supp(h^*) \setminus \supp(g)$.
We write $\vec{\supp}(h^*)$ for the tuple that enumerates $\supp(h^*)$ in the chosen order. For each $h \in O$, there exists $\pi_{h} \in \StabP(\supp(g))$ (which we can choose) such that $h = \pi_{h}(h^*)$. Then define $\vec{\supp}(h) \coloneqq \pi_{h}(\vec{\supp}(h^*)) = ( \pi_{h}(\vec{\supp}(h^*)_1), \dots, \pi_{h}(\vec{\supp}(h^*)_s))$. Note that by \cref{lem:supportsMoveWithPermutation}, $\pi_{h}(\supp(h^*)) = \supp(h)$, so this definition makes sense. Since $\pi_{h}$ fixes $\supp(g)$, also in $\vec{\supp}(h)$, the elements of $\supp(g) \cap \supp(h)$ are the first entries in the tuple. In this way, we have fixed an ordering of the support of each gate in $O$.

For every $h \in O$, we write $\alpha\beta(\bar{x}_{\vec{\supp}(h)}) \in \{0,1\}^{s}$ for the ordered binary string that we obtain by replacing in the tuple $\bar{x}_{\vec{\supp}(h)}  = (x_{\vec{\supp}(h)_1}, \dots, x_{\vec{\supp}(h)_s})$ every entry with its image under $\alpha$ or $\beta$, whichever applies.

Now for every binary string $\gamma \in \{0,1\}^{s}$, let
\[
O_{\alpha\beta, \gamma} \coloneqq \{h \in O \mid \alpha\beta(\bar{x}_{\vec{\supp}(h)}) = \gamma\}.
\]
Let $I = [|\supp(g) \cap \supp(h^*)|]$. As explained above, for every $h \in O$, this is the set of indices of $\vec{\supp}(h)$ that are occupied by elements of $\supp(g) \cap \supp(h)$. Write $\gamma|_I$ for the substring of $\gamma$ at the indices in $I$, and write $\alpha|_I \in \{0,1\}^{|I|}$ for the binary string given by the assignment $\alpha$ restricted to $X_{\supp(g) \cap \supp(h)}$, where the elements are ordered as in $\vec{\supp}(h)$.
\setcounter{lemma}{3} 
\begin{claim}
	\label{claim:Ogamma}
	~\\
	\vspace*{-1em}
	\begin{itemize}
		\item If $\gamma|_I \neq \alpha|_I$, then $O_{\alpha\beta, \gamma} = \emptyset$.
		\item If $\gamma|_I = \alpha|_I$, then every gate $h \in O_{\alpha\beta, \gamma}$ computes the same value $b_{\alpha\beta, \gamma} \in \{0,1\}$ under the assignment $\alpha\beta$, and we have
		\begin{align*}
			b_{\alpha\beta, \gamma} = h_{\vec{\supp}(h) \mapsto \gamma}(|\alpha\beta(\bar{x})|_1 - |\gamma|_1) \text{ for any } h \in O_{\alpha\beta, \gamma}. 
		\end{align*}
	\end{itemize}
\end{claim}
\begin{claimproof}
	If $\gamma|_I \neq \alpha|_I$, then there is no $h \in O$ with $\alpha\beta(\vec{\supp}(h)) = \gamma$.
	Suppose now that $\gamma|_I = \alpha|_I$ and $O_{\alpha\beta, \gamma} \neq \emptyset$. 
	Let $h_1, h_2 \in O_{\alpha\beta, \gamma}$ be arbitrary and let $\sigma \coloneqq \pi_{h_2} \circ \pi_{h_1}^{-1}$. Then $\sigma(h_1) = h_2$, and by \cref{lem:supportsMoveWithPermutation}, $\sigma(\supp(h_1)) = \supp(h_2)$. Moreover, $\sigma$ preserves the ordering of the supports. That is, for every $i \in [s]$, the $i$-th entry in $\vec{\supp}(h_1)$ is mapped by $\sigma$ to the $i$-th entry of $\vec{\supp}(h_2)$ (by our definition of the orderings of the supports).  
	By \cref{lem:symmetricCircuitSemantics} 
	\[
	h_1(\alpha\beta(\bar{x})) = h_2(\alpha\beta(\sigma^{-1}(\bar{x}))).
	\]
	Now $\alpha\beta(\sigma^{-1}(\bar{x}))$ is an ordered binary string that can be reordered via a permutation $\sigma' \in \StabP(\supp(h_2))$ into precisely the string $\alpha\beta(\bar{x})$, that is, $\alpha\beta((\sigma' \circ\sigma^{-1})(\bar{x}))) = \alpha\beta(\bar{x})$. The existence of the desired $\sigma'$ can be seen as follows: Since $h_2 \in O_{\alpha\beta, \gamma}$, we know that $\alpha\beta(\bar{x}_{\vec{\supp}(h_2)}) = \gamma$. And also, $\alpha\beta(\sigma^{-1}(\bar{x})_{\vec{\supp}(h_2)}) = \gamma$ (this was ensured by the choice of $\sigma$). So $\alpha\beta(\bar{x})$ and $\alpha\beta(\sigma^{-1}(\bar{x}))$ agree on the substring indexed by $\supp(h_2)$. Thus, we can choose $\sigma' \in \StabP(\supp(h_2))$ so that it reorders the positions indexed with $[n] \setminus \supp(h_2)$ in such a way that $\alpha\beta((\sigma' \circ\sigma^{-1})(\bar{x}))) = \alpha\beta(\bar{x})$.
	
	Because $h_2$ computes a $\StabP(\supp(h_2))$-symmetric function, applying $\sigma' \in \StabP(\supp(h_2))$ does not change the computed value, so in total, we get
	\[
	h_1(\alpha\beta(\bar{x})) = h_2(\alpha\beta(\sigma^{-1}(\bar{x}))) = h_2(\alpha\beta((\sigma' \circ\sigma^{-1})(\bar{x}))) = h_2(\alpha\beta(\bar{x})).
	\]
	This shows that under the assignment $\alpha\beta$, every gate in $O_{\alpha\beta,\gamma}$ outputs the same value. 
	This value $b_{\alpha\beta, \gamma}$ only depends on the assignment to the variables indexed with the support of a gate $h \in O_{\alpha\beta, \gamma}$, and on the number of $1$s assigned to the variables in $X_{[n] \setminus \supp(h)}$. The assignment to the ordered support $\vec{\supp}(h)$ 
	is the same, i.e.\ $\vec{\supp}(h) \mapsto \gamma$, for each $h \in O_{\alpha\beta, \gamma}$. For every $h \in O_{\alpha\beta,\gamma}$, the number of $1$s assigned to the variables in $X_{[n] \setminus \supp(h)}$ is $|\alpha\beta(\bar{x})|_1 - |\gamma|_1$, which is the total number of $1$s in $\alpha\beta$ minus the ones that are assigned to $X_{\supp(h)}$.
\end{claimproof}

With this claim, we can now write:
\begin{align}
	O_\alpha(\beta(\bar{x}_{[n] \setminus \supp(g)})) = \sum_{\gamma \in \{0,1\}^{s}} |O_{\alpha\beta, \gamma}| \cdot b_{\alpha\beta, \gamma}. \label{eq:sumOverAllGamma}
\end{align}
It remains to determine $|O_{\alpha\beta, \gamma}|$. We write $\gamma_{\setminus I}$ for the substring of $\gamma$ on the positions not in $I$, that is, $\gamma$ without the first $|I|$ symbols. Also, we use $| \cdot |_0$ to denote the number of $0$s in a given binary string, analogously to $| \cdot |_1$ for the number of $1$s.

\begin{claim}
	\label{claim:sizeOfGroupedSets}	
	For every $\gamma \in \{0,1\}^{s}$ such that $\gamma|_I = \alpha|_I$, there exists a constant $\lambda \in \bbN$ (possibly $\lambda = 0$) such that
	\begin{align*}
		|O_{\alpha\beta, \gamma}| = 
		\binom{|\beta(\bar{x}_{[n] \setminus \supp(g)})|_1 }{|\gamma_{\setminus I}|_1} \cdot \binom{|\beta(\bar{x}_{[n] \setminus \supp(g)})|_0 }{|\gamma_{\setminus I}|_0} \cdot \lambda.
	\end{align*}
\end{claim}
\begin{claimproof}
	For each $h \in O$, write $\vec{S}(h)$ for the tuple that lists the elements of $S(h) = \supp(h) \setminus \supp(g)$ in the natural order on $[n]$.
	Now assuming that $\gamma|_I = \alpha|_I$, a gate $h \in O$ is in $O_{\alpha\beta, \gamma}$ iff $\beta(\bar{x}_{\vec{S}(h)}) = \gamma_{\setminus I}$. We count for how many gates this is the case. That is, we have to determine the size of the set
	\[
	A \coloneqq \{ h \in O \mid \beta(\bar{x}_{\vec{S}(h)}) = \gamma_{\setminus I}   \}.
	\]
	For $j \in \{0,1\}$, let $X^j  \coloneqq \{ x_i \in X_{[n]\setminus \supp(g)} \mid \beta(x_i) = j\}$.
	For each pair of sets $(B^0, B^1)$ with $B^0 \subseteq X^0, B^1 \subseteq X^1$ and $|B^0| = |\gamma_{\setminus I}|_0, |B^1| = |\gamma_{\setminus I}|_1$, we define
	\[
	A_{B^0, B^1} \coloneqq \{ h \in A \mid S(h) \cap X^0 = B^0 \text{ and }  S(h) \cap X^1 = B^1    \}.
	\]
	Then the sets $A_{B^0,B^1}$ partition $A$, where $B^0$ ranges over all size-$|\gamma_{\setminus I}|_0$ subsets of $X^0$ and $B^1$ ranges over all size-$|\gamma_{\setminus I}|_1$ subsets of $X^1$. 
	The total number of such pairs $(B^0, B^1)$ is clearly 
	\[
	\binom{|\beta(\bar{x}_{[n] \setminus \supp(g)})|_1 }{|\gamma_{\setminus I}|_1} \cdot \binom{|\beta(\bar{x}_{[n] \setminus \supp(g)})|_0 }{|\gamma_{\setminus I}|_0}.
	\]
	It remains to argue that all parts of this partition have equal size, that is, there is a $\lambda \in \bbN$ such that for each choice of $(B^0, B^1)$, $|A_{B^0, B^1}| = \lambda$.
	
	To this end, let $(B^0, B^1)$ be an arbitrary pair such that $A_{B^0,B^1} \neq \emptyset$. 
	Fix an arbitrary $h \in A_{B^0,B^1}$. For $j \in \{0,1\}$, let $Y^j  \coloneqq \{ i \in S(h) \mid \beta(x_i) = j\}$.
	Let $\Delta(h) \leq \StabP([n] \setminus S(h))$ be the group consisting of all $\pi \in \StabP([n] \setminus S(h))$ that fix the sets $Y^0$ and $Y^1$ setwise. Let $\Stab_\Delta(h) \leq \Delta(h)$ be the subgroup of $\Delta(h)$ that fixes $h$. 
	All gates $h'$ in the $\Delta(h)$-orbit of $h$, denoted $\Orb_\Delta(h)$, satisfy $\beta(\bar{x}_{\vec{S}(h')}) =  \beta(\bar{x}_{\vec{S}(h)}) = \gamma_{\setminus I}$, and $\supp(h') = \supp(h)$, so $S(h') \cap X^j = B^j$ for each $j \in [2]$. Thus, $\Orb_\Delta(h) \subseteq A_{B^0,B^1}$. Conversely, we also have $A_{B^0,B^1} \subseteq \Orb_\Delta(h)$: Let $h' \in A_{B^0,B^1}$ be arbitrary. Then $S(h') = S(h) = B^0 \cup B^1$, for the $h \in  A_{B^0,B^1}$ we fixed before, and thus also $\supp(h) = \supp(h')$. We also know that $\beta(\bar{x}_{\vec{S}(h')}) = \gamma_{\setminus I}$. As $h,h' \in O$, there is a permutation $\pi \in \StabP(\supp(g))$ such that $h' = \pi(h)$. We can assume that $\pi$ fixes every point in $[n] \setminus \supp(h')$ (if not, compose it with a permutation that undoes the action of $\pi$ on $[n] \setminus \supp(h')$; this will fix $h' = \pi(h)$ by the definition of support). So in total, $\pi \in \StabP([n] \setminus S(h))$, and thus it fixes $S(h)$ setwise.
	We can also assume that $\pi(\vec{S}(h))= \vec{S}(h')$ (otherwise, if $\pi$ does not order $S(h')$ according to the order $\vec{S}(h')$ we chose via $\pi_{h'}(\vec{\supp}(h^*))$, then we can compose $\pi$ with a permutation that moves $h'$ to $h^*$ and then back to $h'$ via $\pi_{h'}$; the new permutation $\pi$ thus obtained does satisfy $\pi(\vec{S}(h))= \vec{S}(h')$).
	Hence, $\pi$ must fix the sets $Y^j$ for both $j \in [2]$ because else, we would not have $\beta(\bar{x}_{\vec{S}(h')}) = \beta(\bar{x}_{\vec{S}(h)}) = \gamma_{\setminus I}$.
	In total, this shows that $\pi \in \Delta(h)$, and so, $A_{B^0,B^1} = \Orb_\Delta(h)$.
	
	By the Orbit-Stabiliser Theorem, $|A_{B^0,B^1}| =  |\Orb_\Delta(h)| = \frac{|\Delta(h)|}{|\Stab_\Delta(h)|}$.  
	Now both these group sizes are independent of the choice of $(B^0, B^1)$ because for different gates $h \in A$, the respective groups $\Delta(h)$ and $\Stab_\Delta(h)$ are conjugate to one another.
	More precisely: Take any another pair $(C^0,C^1)$ where $C^0$ is a size-$|\gamma_{\setminus I}|_0$ subset of $X^0$ and $C^1$ a size-$|\gamma_{\setminus I}|_1$ subset of $X^1$. Clearly, there is a permutation $\pi \in \StabP(\supp(g))$ such that $\pi(B^0) = C^0$ and $\pi(B^1) = C^1$. Then, for our fixed gate $h \in A_{B^0,B^1}$, we have $h' \coloneqq \pi(h) \in A_{C^0,C^1}$ (by \cref{lem:supportsMoveWithPermutation}). Thus, we must also have $\Delta(h') = \pi \Delta(h) \pi^{-1}$ and $\Stab_\Delta(h') = \pi \Stab_\Delta(h) \pi^{-1}$. Therefore, $|A_{C^0,C^1}| = |\Orb_\Delta(h')| =  |\Orb_\Delta(h)| = |A_{B^0,B^1}|$, which is what we had to show.
	
	In total, we either have $\lambda = 0$, or, as long as there is at least one non-empty set $A_{B^0,B^1}$, they all have the same positive cardinality $\lambda > 0$.
\end{claimproof}

We finally put everything together:
\begin{align*}
	&O_\alpha(\beta(\bar{x}_{[n] \setminus \supp(g)})) = \sum_{\gamma \in \{0,1\}^{|\supp(h^*)|}} |O_{\alpha\beta, \gamma}| \cdot b_{\alpha\beta, \gamma}\\
	&=\sum_{\stackrel{\gamma \in \{0,1\}^{|\supp(h^*)|}}{\gamma|_I = \alpha|_I}} |O_{\alpha\beta, \gamma}| \cdot b_{\alpha\beta, \gamma}\\
	&= \sum_{\stackrel{\gamma \in \{0,1\}^{|\supp(h^*)|}}{\gamma|_I = \alpha|_I}} \Big( \binom{|\beta(\bar{x}_{[n] \setminus \supp(g)})|_1 }{|\gamma_{\setminus I}|_1} \cdot \binom{|\beta(\bar{x}_{[n] \setminus \supp(g)})|_0 }{|\gamma_{\setminus I}|_0} \cdot \lambda\\
	&\cdot h^{\gamma}_{\vec{\supp}(h) \mapsto \gamma}(|\alpha\beta(\bar{x})|_1 - |\gamma|_1) \Big), \tag{$\star$}
\end{align*}
where $h^{\gamma} \in O_{\alpha\beta,\gamma}$ denotes an arbitrarily chosen representative of that set. The first equality is \eqref{eq:sumOverAllGamma}, the second equality is due to the first part of \cref{claim:Ogamma}, and the third equality is given by \cref{claim:sizeOfGroupedSets}, together with the second part of \cref{claim:Ogamma}.

By the induction hypothesis, we know that the $\StabP(\supp(h^\gamma))$-symmetric function $h^{\gamma}_{\vec{\supp}(h) \mapsto \gamma}$ has a period of length 
$
m \cdot \prod_{i \in [r]} p_i^{c_i},
$
where each exponent satisfies $c_i \leq \lfloor \log_{p_i}(s(n)) \rfloor$.
Now it remains to determine from the above expression an upper bound on the period length for the function $O_\alpha(\bar{x}_{[n] \setminus \supp(g)}) \mod m$. 

\begin{claim}
	\label{claim:periodOfBinomialCoefficient}	
	For fixed $\gamma \in \{0,1\}^{|\supp(h)|}$ with $\gamma|_I = \alpha|_I$, the period length of 
	\begin{align*}
		&\binom{|\beta(\bar{x}_{[n] \setminus \supp(g)})|_1 }{|\gamma_{\setminus I}|_1} \cdot  \binom{|\beta(\bar{x}_{[n] \setminus \supp(g)})|_0 }{|\gamma_{\setminus I}|_0} \cdot \lambda \mod m,
	\end{align*}
	with respect to the value of $|\beta(\bar{x}_{[n] \setminus \supp(g)})|_1$, 
	is at most $
	m \cdot \prod_{i \in[r]} p_i^{c_i},
	$
	where each exponent satisfies $c_i \leq \lfloor \log_{p_i}(s(n)) \rfloor$.
\end{claim}
\begin{claimproof}
	We have:
	\begin{align*}
		&\binom{|\beta(\bar{x}_{[n] \setminus \supp(g)})|_1 }{|\gamma_{\setminus I}|_1} \cdot \binom{|\beta(\bar{x}_{[n] \setminus \supp(g)})|_0 }{|\gamma_{\setminus I}|_0} \cdot \lambda \mod m\\
		= & \binom{|\beta(\bar{x}_{[n] \setminus \supp(g)})|_1 }{|\gamma_{\setminus I}|_1} \cdot \binom{n-|\supp(g)| -|\beta(\bar{x}_{[n] \setminus \supp(g)})|_1 }{|\gamma_{\setminus I}|_0} \cdot \lambda \mod m
	\end{align*}
	For analysing the periodic behaviour, the constant factor $\lambda$ is irrelevant and can be dropped. The constants $|\gamma_{\setminus I}|_0$ and $|\gamma_{\setminus I}|_1$ are between $0$ and $|\supp(h) \setminus \supp(g)| \leq s(n)$ (recall that $s(n)$ is the upper bound on the support size we assume in \cref{lem:inductionPeriodLength}). 
	By \cref{thm:binomialPeriod}, each of the binomial coefficients has some period length modulo $m$ (with respect to $|\beta(\bar{x}_{[n] \setminus \supp(g)})|_1$), not necessarily the same one. But in both cases, the period length is of the form 
	$
	m \cdot \prod_{i \in [r]} p_i^{c_i},
	$
	where each exponent satisfies $c_i \leq \lfloor \log_{p_i}(s(n)) \rfloor$.
	
	Note that in the second binomial coefficient, $|\beta(\bar{x}_{[n] \setminus \supp(g)})|_1$ appears with a negative sign, but this does not change the periodicity of the binomial coefficient with respect to this value.
	Also note that $n-|\supp(g)| -|\beta(\bar{x}_{[n] \setminus \supp(g)})|_1$ is between $0$ and $n-|\supp(g)|$, just like $|\beta(\bar{x}_{[n] \setminus \supp(g)})|_1$.   
	Hence, the period length of the product of the two binomial coefficients is the least common multiple of their period lengths. This is obtained by taking for each prime factor $p_i$ of $m$ the greater of its two exponents appearing in the period lengths of the two binomial coefficients. Thus, the period length of the product is again of the form $m \cdot \prod_{i \in[r]} p_i^{c_i}$ with $c_i \leq \lfloor \log_{p_i}(s(n)) \rfloor$.
\end{claimproof}

Now we can combine this claim with the induction hypothesis to compute the period of $(\star)$, when taken modulo $m$.
Since $|\gamma|_1$ is a fixed constant, and $\alpha$ is fixed, the period length of $h^{\gamma}_{\vec{\supp}(h) \mapsto \gamma}(|\alpha\beta(\bar{x})|_1 - |\gamma|_1)$, when viewed as a function of $|\beta(\bar{x}_{[n] \setminus \supp(g)})|_1$, is indeed as given by the induction hypothesis, namely of the form
$
m \cdot \prod_{i \in [r]} p_i^{c_i}.
$
Like in the preceding claim, each exponent satisfies $c_i \leq \lfloor \log_{p_i}(s(n)) \rfloor$.
Thus, together with that claim, it follows that each summand in $(\star) \mod m$ also has a period length of this form: It is again the least common multiple of all the occurring period lengths, which is given by taking the respective greatest exponent $c_i$, for every $i \in [r]$, that appears. 
For the same reason, the period length of the whole sum, which is $O_\alpha(\beta(\bar{x}_{[n] \setminus \supp(g)})) \mod m$, is of this form.

This finishes the period estimation for one orbit $O$. The period length of $g_\alpha$ is the least common multiple of the period lengths of the $O_\alpha$, where $O$ ranges over all $\StabP(\supp(g))$-orbits that $gE(C)$ is partitioned into.
So in total, the period length of $g_\alpha$ is of the form $m \cdot \prod p_i^{c_i}$, where each $c_i$ is at most $\lfloor \log_{p_i}(s(n)) \rfloor$.
This finishes the inductive step in the proof of \cref{lem:inductionPeriodLength}.

\section{Size lower bound for nested block symmetry}

In this section, fix $h \in \bbN$ and a tuple $\boldsymbol{k} = (k_1(n), \dots, k_h(n))$ such that $\prod_{i \in [h]} k_i(n) = n$ for each $n \in \bbN$.
For every $n \in \bbN$, let 
\begin{align*}
k_{\text{min}}(n) \coloneqq \min_{i \in [h]} k_i(n).\\
k_{\text{max}}(n) \coloneqq \max_{i \in [h]} k_i(n).
\end{align*}
denote the smallest and largest block sizes in the tree $\Tt_n^{\boldsymbol{k}}$. 

We now show how to adapt the proof from the previous section to obtain the circuit size lower bound claimed in \cref{thm:main2} for $\Aut(\Tt_n^{\boldsymbol{k}})$-symmetric circuits. 
The main technical result from which the lower bound can be derived is the following variation of \cref{thm:mainSupportLowerBound}:

\begin{theorem}
	\label{thm:mainSupportLowerBoundNested}
	Fix a positive integer $m > 3$ and let $r$ be the number of distinct prime divisors of $m$.
	Let $(C_n)_{n \in \bbN}$ be a family of $\Aut(\Tt_n^{\boldsymbol{k}})$-symmetric $\MOD_m$-circuits. 
	Let $B \in \Bb(\Tt_n^{\boldsymbol{k}})$ be a block that has size $|B| = k_j(n)$, for some $j \in [h]$. \footnote{ Note that we can fix the block $B$ independently of $n$ since the structure of $\Tt_n^{\boldsymbol{k}}$ only depends on $\boldsymbol{k}$, and $n$ just controls the size of $B$.} If $\maxSup_{B}(C_n) < (k_j(n)/m)^{1/r}$, then $C_n$ does not compute $\AND_n$.
\end{theorem}

\begin{corollary}
	\label{cor:sizeLowerBoundNested}
	In the setting of the theorem,
	if $C_n$ computes $\AND_n$ for an $n$ such that $k_{\text{min}}(n) > 8$, then
	$|V(C_n)| \geq \binom{k_{\text{max}}(n)}{(k_{\text{max}}(n)/m)^{1/r}}$.
\end{corollary}
\begin{proof}
	If $C \coloneqq C_n$ computes $\AND_n$, then by \cref{thm:mainSupportLowerBoundNested},
	for every block $B \in \Bb(\Tt_n^{\boldsymbol{k}})$, $\maxsup_B(C) \geq (|B|/m)^{1/r}$.
	Then for every $B \in \Bb(\Tt_n^{\boldsymbol{k}})$, $\maxorb_{\Stab(B)}(C) \geq \binom{|B|}{ (|B|/m)^{1/r} }$ by \cref{lem:supportSizes} (2). 
	Since $|\maxorb_{\Stab(B)}(C)| \leq |V(C)|$, for every $B \in  \Bb(\Tt_n^{\boldsymbol{k}})$, we can pick a block of maximal size and obtain $|V(C)| \geq \binom{k_{\text{max}}(n)}{(k_{\text{max}}(n)/m)^{1/r}}$.
\end{proof}	
Just like in the last section, the size lower bound from this corollary translates into
\[
|V(C_n)| \geq 2^{\Omega(k_{\text{max}}(n)^{1/r} \cdot \log(k_{\text{max}}(n)) )},
\]
which is what is stated in \cref{thm:main2}.
Similarly as in the previous section, \cref{thm:mainSupportLowerBoundNested} is proved by showing that if the supports are not big enough, then the functions computed by the gates have a certain periodic behaviour. However, the notion of period is different now because it has to match the different notion of symmetry. 

Recall from Section \ref{sec:permGroups} that for a tree $\Tt^{\boldsymbol{k}}_n$, and a set $W \subseteq V(\Tt^{\boldsymbol{k}}_n)$ of nodes, $L_0(W)$ denotes the set of leaves in subtrees rooted at nodes in $W$.

\begin{definition}[Block-periodic functions]
	\label{def:nestedBlockPeriod}
	Let $B \in \Bb(\Tt_n^{\boldsymbol{k}})$. Let $\Gamma \leq \Aut(\Tt_n^{\boldsymbol{k}})$ be a group such that for every $\pi \in \Sym(B)$, there is a $\sigma \in \Gamma$ that fixes $B$ setwise and satisfies $\sigma|_B = \pi$.
	Let $f(x_1, \dots, x_n)$ be a $\Gamma$-symmetric function.  
	Let $\beta \colon \{x_1, \dots, x_n\} \to \{0,1\}$ be an assignment such that for each $v \in B$, $\beta(\bar{x}_{L_0(v)}) \in \{\bar{0}, \bar{1}\}$, i.e., $\beta$ is constant on each set $X_{L_0(v)}$, for each $v \in B$. 
	 Then by $\Gamma$-symmetry of $f$, the value of $f(\beta(\bar{x}))$ depends only on $\beta(\bar{x}_{[n] \setminus L_0(B)})$ and on the number 
	\[
	|\beta(\bar{x})|_{1}^B \coloneqq   |\{v \in B \mid  \beta(\bar{x}_{L_0(v)}) = \bar{1}  \}|.
	\]
	We say that $f$ has a \emph{$B$-period} of length $\ell$ if 
	\[
	f(\beta(\bar{x})) = f(\beta'(\bar{x})),
	\]
	for any two assignments $\beta, \beta'$ that are constant on $X_{L_0(v)}$, for each $v \in B$, and satisfy  $|\beta'(\bar{x})|_{1}^B = |\beta(\bar{x})|_{1}^B+\ell$, and $\beta(\bar{x}_{[n] \setminus L_0(B)}) = \beta'(\bar{x}_{[n] \setminus L_0(B)}) \in \{\bar{0}, \bar{1}\}$.  
\end{definition}

\begin{lemma}
	\label{lem:nestedBlockPeriodNotAnd}
	Let $B, \Gamma$ and $f$ be as in \cref{def:nestedBlockPeriod}. 
	If $f(x_1, \dots, x_n)$ has a $B$-period of length $1 \leq \ell \leq |B|$, then $f \neq \AND_n$.
\end{lemma}	
\begin{proof}
	Suppose for a contradiction that $f = \AND_n$. Then $f(\bar{1}) = 1$. Consider an assignment $\beta \colon \{x_1, \dots, x_n\} \to \{0,1\}$, which is $1$ everywhere except that precisely for $\ell$ nodes $v \in B$, $\beta(\bar{x}_{L_0(v)}) = \bar{0}$. 
	Then $f(\beta(\bar{x})) = 1$ because $f$ has a $B$-period of length $\ell$. But this is a contradiction because $\AND_n(\beta(\bar{x})) \neq 1$. 
\end{proof}	

\begin{proof}[Proof of \cref{thm:mainSupportLowerBoundNested}]
	Assume the setting of \cref{thm:mainSupportLowerBound}, so in particular, $\maxSup_{B}(C_n) < (|B|/m)^{1/r}$ for some block $B \in \Bb(\Tt_n^{\boldsymbol{k}})$. 
	By \cref{cor:upperboundPeriodLengthNested} below, the $\Aut(\Tt_n^{\boldsymbol{k}})$-symmetric function computed by $C_n$ has a $B$-period of length at most $m \cdot ((|B|/m)^{1/r})^r = |B|$. But then, this function cannot be $\AND_n$ by \cref{lem:nestedBlockPeriodNotAnd}.
\end{proof}	

Again, it remains to prove the key technical ingredient, \cref{cor:upperboundPeriodLengthNested}, and again, this requires to adjust the notion of periodicity to stabiliser groups of gates that fix the support pointwise.

Let $S \subseteq B$, let $\Gamma \leq \Aut(\Tt_n^{\boldsymbol{k}})$ be a group such that for every $\pi \in \StabP_{\Sym(B)}(S)$, there is a $\sigma \in \Aut(\Tt_n^{\boldsymbol{k}})$ that fixes $B$ setwise and satisfies $\sigma|_B = \pi$. We now consider $\StabP_\Gamma(S) \leq \Gamma$, the pointwise stabiliser of $S$ in $\Gamma$. 
Let $f(x_1, \dots, x_n)$ be a $\StabP_\Gamma(S)$-symmetric function.
Fix an assignment $\alpha \colon X_{[n] \setminus L_0(B \setminus S)} \to \{0,1\}$ which is constant on each set $X_{L_0(v)}$ for each $v \in S$, and constant on $X_{[n] \setminus L_0(B)}$.
Now when $\alpha$ is regarded as fixed, then for assignments $\beta \colon X_{L_0(B \setminus S)} \to \{0,1\}$ that are constant on each set $X_{L_0(v)}$ for each $v \in B \setminus S$, the value of $f(\alpha\beta(\bar{x}))$ only depends on $|\beta(\bar{x})|_{1}^B$.    

Analogously to the previous section, we write $f_\alpha \colon X_{L_0(B \setminus S)} \to \{0,1\}$ for the function obtained from $f$ by fixing the variables in $X_{[n] \setminus L_0(B \setminus S)}$ to the values given by $\alpha$. That is, \[
f_\alpha(\bar{x}_{L_0(B \setminus S)}) \coloneqq f(\alpha(\bar{x}_{[n] \setminus L_0(B \setminus S)})\bar{x}_{L_0(B \setminus S)}).
\]
When we say that a $\StabP_\Gamma(S)$-symmetric function $f$ has a $B$-period, we mean that for every $\alpha \colon  X_{[n] \setminus L_0(B \setminus S)}  \to \{0,1\}$ that is constant on $X_{L_0(v)}$ for each $v \in S$, and constant on $X_{[n] \setminus L_0(B)}$, the function $f_\alpha(\bar{x}_{L_0(B \setminus S)})$ has a $B$-period. 

Now let $g \in V(C)$ be a gate and let $\Gamma \coloneqq \Stab_{\Aut(\Tt_n^{\boldsymbol{k}})}(g)$. 
By the definition of $\supp_{B}(g)$, we know that for every $\pi \in \StabP_{\Sym(B)}(S)$, there is a $\sigma \in \Stab_{\Aut(\Tt_n^{\boldsymbol{k}})}(g)$ that fixes the set $B$ setwise and satisfies $\sigma|_B = \pi$. Therefore, $\Gamma$ has the properties we are assuming in the above paragraph.
The function $g(\bar{x})$ computed by $g$ is always $\Gamma$-symmetric by \cref{lem:symmetricCircuitSemantics}. Because $\StabP_\Gamma(\supp_{B}(g)) \leq \Gamma$, it is also $\StabP_\Gamma(\supp_{B}(g))$-symmetric. So when we say that $g(\bar{x})$ has a $B$-period of length at most $\ell$, we mean that $g_\alpha(\bar{x}_{L_0(B \setminus S)})$ has such a period for every $\alpha \colon X_{[n] \setminus L_0(B \setminus \supp_{B}(g))} \to \{0,1\}$ that is constant on each $X_{L_0(v)}$ for each $v \in \supp_{B}(g)$, and constant on $X_{[n] \setminus L_0(B)}$.
Now the technical lemma that we want to show reads as follows.

\begin{restatable}{lemma}{inductionPeriodLengthNested}
	\label{lem:inductionPeriodLengthNested}
	Let $B \in \Bb(\Tt_n^{\boldsymbol{k}})$ be a block of size $k_j(n)$, for some $j\in [h]$. Let $s \colon \bbN \to \bbN$ be a function in $o(n)$. Fix a number $m \in \bbN$.
	Let $(C_n)_{n \in \bbN}$ be a family of $\Aut(\Tt_n^{\boldsymbol{k}})$-symmetric rigid $\MOD_m$-circuits such that $\maxsup_{B}(C_n) < s(k_j(n))$ for all $n \in \bbN$.
	Let $g$ be a gate in $C_n$ and let $\Gamma \coloneqq  \Stab_{\Aut(\Tt_n^{\boldsymbol{k}})}(g)$. Then the $\StabP_\Gamma(\supp_{B}(g))$-symmetric function $g(\bar{x})$ has a $B$-period of length at most 
	\[
	q(m,s) \coloneqq m \cdot \prod_{i \in [r]} p_i^{\lfloor \log_{p_i}(s(k_j(n))) \rfloor},
	\]
	where the product ranges over the prime factors $p_1, \dots, p_r$ of $m$.
\end{restatable}

\begin{corollary}
	\label{cor:upperboundPeriodLengthNested}
	In the setting of \cref{lem:inductionPeriodLengthNested}, the $\Aut(\Tt_n^{\boldsymbol{k}})$-symmetric function computed at the output gate of $C_n$ has a $B$-period of length at most $m \cdot s(k_j(n))^r$, where $r$ denotes the number of distinct prime divisors of $m$.
\end{corollary}

With this, we can state the proof of the lower bound theorem.
\begin{proof}[Proof of \cref{thm:mainSupportLowerBoundNested}]
	Assume the setting of \cref{thm:mainSupportLowerBoundNested}, so in particular, $\maxSup_{B}(C_n) < (|B|/m)^{1/r}$ for some block $B \in \Bb(\Tt_n^{\boldsymbol{k}})$. 
	By \cref{cor:upperboundPeriodLengthNested}, the $\Aut(\Tt_n^{\boldsymbol{k}})$-symmetric function computed by $C_n$ has a $B$-period of length at most $m \cdot ((|B|/m)^{1/r})^r = |B|$. But then, this function cannot be $\AND_n$ by \cref{lem:nestedBlockPeriodNotAnd}.
\end{proof}	

It remains to prove the technical core, \cref{lem:inductionPeriodLengthNested}.
This is done analogously to the proof of \cref{lem:inductionPeriodLength}, where we essentially ``zoom in'' on the group $\Sym(B)$ instead of performing the calculation for the symmetry group $\Sym_n$. 
We refrain from reiterating the proof of \cref{lem:inductionPeriodLength} with all technicalities. Instead, we only highlight the differences between the two settings.

The goal is again to prove for every gate $g \in V(C)$: For every assignment 
\[
\alpha \colon X_{[n] \setminus L_0(B \setminus \supp_{B}(g))} \to \{0,1\}
\]
that is constant on each $X_{L_0(v)}$ for each $v \in \supp_{B}(g)$, and constant on $X_{[n] \setminus L_0(B)}$, the function $g_\alpha(\bar{x}_{L_0(B \setminus S)})$ has a $B$-period of length  $1$ or of the form 
\[
m \cdot \prod_{i \in [r]} p_i^{c_i},
\]
where each exponent $c_i$ satisfies $c_i \leq \lfloor \log_{p_i}(s(n)) \rfloor$. In the inductive step of the proof, we consider a fixed gate $g$ and assume this statement holds for each of its children. Let $\Gamma \coloneqq  \Stab_{\Aut(\Tt_n^{\boldsymbol{k}})}(g)$.
Now we partition the children into $\StabP_\Gamma(\supp_{B}(g))$-orbits. For such an orbit $O \subseteq gE(C)$,we define $O_\alpha(\bar{x}_{L_0(B \setminus S)})$ analogously as before, as the sum over the evaluations of all $h \in O$, with this fixed partial assignment $\alpha$. 
We then consider assignments $\beta \colon X_{L_0(B \setminus S)} \to \{0,1\}$ but only those that are constant on $X_{L_0(v)}$, for each $v \in B \setminus S$. In the same way as in the proof of \cref{lem:inductionPeriodLength}, we define an ordered support tuple $\vec{\supp}_{B}(h)$ of the same length $s$ for each $h \in O$. Then for every binary string $\gamma \in \{0,1\}^s$, we let
\[
O_{\alpha\beta,\gamma} \coloneqq \{h \in O \mid \alpha\beta(\bar{x}_{L_0(\vec{\supp}_{B}(h))}) = \boldsymbol{\gamma}\},
\]
where $\boldsymbol{\gamma} \in \{0,1\}^{s \cdot |L_0(v)|}$, for an arbitrary $v \in B$, denotes the ``inflated'' string which arises from $\gamma$ by replacing every symbol $b \in \{0,1\}$ in $\gamma$ with the string $bb \dots b$ of length $|L_0(v)|$.
Similarly as before, let $I = [|\supp_{B}(g) \cap \supp_{B}(h^*)|]$, where $h^* \in O$ is the gate that was chosen to fix the orderings of the supports. This is the set of indices of each $\vec{\supp}_{B}(h)$ that are occupied by elements of $\supp_{B}(g) \cap \supp_{B}(h)$. Write $\gamma|_I$ for the substring of $\gamma$ at the indices in $I$, and write $\alpha|_I \in \{0,1\}^{|I|}$ for the binary string given by the assignment $\alpha$ restricted to $X_{L_0(\supp_{B}(g) \cap \supp_{B}(h))}$, where the elements are ordered as in $\vec{\supp}_{B}(h)$.

The analogue of \cref{claim:Ogamma} is proved in the same way as in the last section, where instead of choosing $\sigma' \in \StabP(\supp(h_2))$, here we have to choose a $\sigma' \in \StabP_{\Stab(h_2)}(\supp_{B}(h_2))$ to reorder the respective string. This is possible because we know by the definition of $\supp_{B}(h_2)$ that every permutation $\pi \in \StabP_{\Sym(B)}(\supp_{B}(h_2))$ is realised by some permutation in $\Aut(\Tt_n^{\boldsymbol{k}})$ that fixes the gate $h_2$; we do not have control over how $\sigma'$ permutes the indices in $L_0(v)$, for each $v \in B$, and in $L_0 \setminus L_0(B)$, but this is irrelevant because $\alpha$ and $\beta$ are constant on these.

Thus, the equation \eqref{eq:sumOverAllGamma} that expresses $O_{\alpha}(\beta(\bar{x}_{L_0(B \setminus S)}))$ in terms of the sizes of the sets $O_{\alpha\beta,\gamma}$ also holds in the setting of nested block symmetry we consider here.

Analogously to \cref{claim:sizeOfGroupedSets}, we now have to show that 
\begin{align*}
	|O_{\alpha\beta, \gamma}| = 
	\binom{|\beta(\bar{x}_{[n] \setminus L_0(\supp_{B}(g))})|_1^B }{|\gamma_{\setminus I}|_1} \cdot \binom{|\beta(\bar{x}_{[n] \setminus L_0(\supp_{B}(g))})|_0^B }{|\gamma_{\setminus I}|_0} \cdot \lambda,
\end{align*}
for some constant $\lambda \in \bbN$. This is shown in the same way as in the proof of \cref{claim:sizeOfGroupedSets}, with the following modifications:
For $j \in \{0,1\}$, we now let $X^j \coloneqq \{  v \in B \setminus \supp_B(g) \mid \beta(\bar{x}_{L_0(v)}) = \bar{j}   \}$. Then the sets $A_{B^0,B^1}$ for pairs $(B^0, B^1)$ with $B^0 \subseteq X^0, B^1 \subseteq X^1$ and $|B^0| = |\gamma_{\setminus I}|_0, |B^1| = |\gamma_{\setminus I}|_1$, are defined as before, where $S(h) \subseteq B$ now denotes the set $\supp_B(h) \setminus \supp_B(g)$.
Then again, we can show that for each choice of $(B^0,B^1)$, the set $A_{B^0,B^1}$ has the same size. This is done as in the proof of \cref{claim:sizeOfGroupedSets}, where the group $\Delta(h)$ is now defined as the subgroup of $\Stab_{\Aut(\Tt^{\boldsymbol{k}}_n)}(g)$ that fixes every $v \in B \setminus (Y^0 \cup Y^1)$, and stabilises each $Y^j$ setwise, for $j \in [2]$. Here, $Y^j$ is defined analogously as before, as the set of all $v \in S(h) \subseteq B$ such that $\beta(\bar{x}_{L_0(v)})$ is constantly $j$. The rest of the reasoning is analogous as in the proof of \cref{claim:sizeOfGroupedSets}.

Finally, the proof of \cref{claim:periodOfBinomialCoefficient}, which shows the periodicity of the above expression for $|O_{\alpha\beta, \gamma}|$, only depends on properties of binomial coefficients. This works completely analogously as in the previous section. Note that here, $|\gamma_{\setminus I}|_1$ and $|\gamma_{\setminus I}|_0$ are both upper-bounded by $\maxSup_B(C) \leq s(|B|) = s(k_j(n))$.
This finishes the proof (sketch) of \cref{lem:inductionPeriodLengthNested}.

\section{Symmetric circuit upper bounds}
\label{sec:upperbound}

In this section, we present the upper bound constructions matching the lower bounds in \cref{thm:main1} and \cref{thm:main2}. The fully symmetric depth-2 construction for \cref{thm:main1} is entirely due to \cite{IdziakKK22LICS} and we include a summary of their proof for completeness in Section \ref{sec:upperbound1}. 
The construction for \cref{thm:main2} in Section \ref{sec:upperbound2} does not appear elsewhere in the literature but it is simply a recursive application of the depth-2 construction. In \cite{IdziakKK22LICS}, a similar but more sophisticated recursive construction is presented, leading to a smaller depth at the cost of greater asymptotic size. Our construction here is a more naive variant of this, which is possibly not depth-optimal, but matches the size lower bound from \cref{thm:main2}.

\subsection{Fully symmetric depth-2 construction}
\label{sec:upperbound1}
\begin{theorem}[{\cite[Proposition 3.1]{IdziakKK22LICS}}]
\label{thm:upperbound1}
   Fix $m\in \mathbb{N}$ with at least $r \geq 2$ distinct prime divisors of $m$. For every $n \in \bbN$, there is a $\Sym_n$-symmetric depth-2 $\MOD_m$-circuit with $2^{\Oo(n^{1/r} \cdot \log n)}$ gates which computes $\AND_n$.
\end{theorem}
This symmetric construction was first provided for depth-$3$ circuits by Barrington, Beigel and Rudrich \cite{BarringtonBR94}, and improved to depth $2$ by Idziak, Kawałek, Krzaczkowski \cite{IdziakKK22LICS}, and independently by Chapman and Williams \cite{ChapmanW22}. 
We outline the proof from \cite{IdziakKK22LICS}.
The main building block is what the authors call \emph{$\bbZ_{pq}$-expressions}. 

\begin{definition}[$\bbZ_{pq}$-expressions]
	Let $p,q$ be two distinct primes. Let $n \in \bbN$. 
	Let $b \colon \bbZ_p \to \bbZ_q$ be the function that maps $0$ to $0$ and every $x \in \bbZ_p \setminus \{0\}$ to $1 \in \bbZ_q$. 
	An $n$-ary \emph{$\bbZ_{pq}$-expression} is of the form
	\[
	\sum_{\beta \in \bbZ_p^n, c \in \bbZ} \alpha_{\beta,c} \cdot b\Big(\sum_{i =1}^n \beta_ix_i +c   \mod p \Big) \mod q,
	\]
	where the coefficients $\alpha_{\beta, c}$ are in $\bbZ_q$, the outer sum and the multiplications with the $\alpha_{\beta, c}$ are evaluated in $\bbZ_q$, while the expression inside $b$ is evaluated in $\bbZ_p$. 
\end{definition}

It is straightforward to see that $\bbZ_{pq}$-expressions can be computed by modular counting circuits of depth $2$, in a certain sense:
Since the output of any $\MOD_m^R$-gate is always Boolean, and the result of a $\bbZ_{pq}$-expression is in $\bbZ_q$, the only way in which we can realise such expressions as $\MOD_m$-circuits is to have multiple output wires. The semantics is that the sum over the output wires modulo $q$ is equal to the result of the $\bbZ_{pq}$-expression. Such a $\MOD_m$-circuit with a set of designated output wires that are to be interpreted as a sum in $\bbZ_q$ is called a $\MOD_m$-circuit of \emph{output type $q$} henceforth.
The depth of a circuit of output type $q$ refers to the maximum number of wires along any path, so the layer consisting of the output wires counts towards the depth.
With this definition it is straightforward to write $\bbZ_{pq}$-expressions as modular circuits:
\begin{lemma}
	\label{lem:fromExpressionsToCircuits}
	Let $p,q,m \in \bbN$ be integers such that $m$ has $p$ and $q$ as prime factors. Every $\bbZ_{pq}$-expression can be realised by a depth-$2$ $\MOD_m$-circuit of output type $q$.
\end{lemma}	
\noindent
For a Boolean assignment $\beta$ to variables $x_1, \dots, x_n$, we write $\beta(\bar{x})$ for the tuple $(\beta(x_1), \dots, \beta(x_n))$. By $|\beta(\bar{x})|_0$, we denote the number of $0$s in this tuple. 
A particular $\bbZ_{pq}$-expression that is central for the proof of \cref{thm:upperbound1} is the following.
\begin{lemma}
	\label{lem:symmetricCircuitsForZpqExpressions}
	Let $\nu \in \bbN$ and let $q$ be a prime.
	The function $t_{q^\nu}(x_1, \dots, x_n)$ which satisfies for all $\beta \colon \{x_1, \dots, x_n\} \to \{0,1\}$:
	\[
	t_{q^\nu}(\beta(\bar{x})) \coloneqq \begin{cases}
												0 & \text{ if } |\beta(\bar{x})|_0  \text{ is divisible by } q^\nu\\
												1 & \text{ else}
		\end{cases}
	\]
	is expressible as a $\bbZ_{pq}$-expression, for every prime $p \neq q$.
	Moreover, for every $m \in \bbN$ that has $p$ and $q$ as prime factors, this $\bbZ_{pq}$-expression can be realised as a depth-$2$ $\Sym_n$-symmetric circuit of output type $q$ and of size at most $2^{\Oo(q^\nu \cdot \log n)}$.
\end{lemma}	
\begin{proof}
	This follows from \cite[Lemma 3.5]{IdziakKK22LICS}. The fact that the circuit can be realised in a $\Sym_n$-symmetric way is not stated explicitly there, but can be seen by inspection of the proof.
	To be precise, the proof of \cite[Fact 3.4]{IdziakKK22LICS} shows that $t_{q^\nu}$ is effectively expressed as a linear combination of elementary symmetric polynomials. Each of these polynomials is by definition $\Sym_n$-symmetric, and this carries over to the $\bbZ_{pq}$-expression representing them.
\end{proof}

To prove the upper bound result, we summarise the proof of \cite[Proposition 3.1]{IdziakKK22LICS}:
\begin{proof}[Proof of \cref{thm:upperbound1}]
	Let $p_1, \dots, p_r$ be the prime factors of $m$. Fix integers $\nu_1, \dots, \nu_r$ such that for each $j \in [r]$, we have $p_j^{\nu_j-1} \leq n^{1/r} < p_j^{\nu_j}$. Let
	\[
	T(\bar{x}) \coloneqq \sum_{j =1}^r \frac{m}{p_j} \cdot t_{p_j^{\nu_j}}(\bar{x}) \mod m.
	\]
	One can show that $T(\beta(\bar{x})) = 0$ if and only if $\beta(x_i) = 1$ for every $i \in [n]$:
	If all $\beta(x_i)$ are equal to $1$, then $|\beta(\bar{x})|_0 = 0$ is divisible by every prime power, so each $t_{p_j^{\nu_j}}$ will evaluate to $0$, and hence $T(\beta(\bar{x})) = 0$. Conversely, assume that $T(\beta(\bar{x})) = 0$. This can only be the case if for all $j \in [r]$,  $t_{p_j^{\nu_j}}(\beta(\bar{x})) = 0$. Then $|\beta(\bar{x})|_0$ is divisible by $\prod_{j \in [r]} p_j^{\nu_j} > n$. Since $|\beta(\bar{x})|_0 \leq n$, it follows that $|\beta(\bar{x})|_0 = 0$, which is what we had to show.
	
	By \cref{lem:symmetricCircuitsForZpqExpressions}, each $t_{p_j^{\nu_j}}$ can be expressed as a depth-2 $\Sym_n$-symmetric $\MOD_m$-circuit $C_j$ of output type $p_j$ and of size at most $2^{\Oo(n^{1/r} \cdot \log n)}$.
	Thus, to compute $\AND_n$, we connect the outgoing wires of the depth-2 symmetric circuits $C_1, \dots, C_r$ to an output gate $\MOD_m^{\{0\}}$ that sums up the values $t_{p_j^{\nu_j}}$ modulo $m$, with the respective coefficients $\frac{m}{p_j}$ realised by appropriate wire multiplicities, and outputs $1$ if and only if $T(\beta(\bar{x})) = 0$. Let this circuit be $C$.
	
	Recall that we defined the depth of a modular circuit of output type $q$ in such a way that it includes its outgoing wires. 
	Thus, the outgoing wires of the $C_j$ are already accounted for in their depth, and adding one more output gate on top does not increase the depth of the resulting circuit. Hence, $C$ also has depth $2$.
\end{proof}

\subsection{Nested block-symmetric construction}
\label{sec:upperbound2}

Now we present the construction that achieves the upper bound in \cref{thm:main2}. It simply applies \cref{thm:upperbound1} to recursively compute the AND over each block defined by the tree $\Tt_n^{\boldsymbol{k}}$. 
\begin{theorem}
	\label{thm:upperbound2}
	Let $m \in \bbN$ be a number with $r \geq 2$ distinct prime factors.
	Fix an $h \in \bbN$ and an $h$-tuple $\boldsymbol{k} = (k_1(n), \dots, k_h(n))$ such that $\prod_{i \in [h]} k_i(n) = n$ for all $n \in \bbN$. 
	Let $k_{\text{max}}(n) \coloneqq \max_{i \in [h]} k_i(n)$.
	For every $n \in \bbN$ there is an $\Aut(\Tt_n^{\boldsymbol{k}})$-symmetric $\MOD_m$-circuit $C_n$ of size $2^{\Oo(k_{\text{max}}(n)^{1/r} \cdot \log k_{\text{max}}(n))}$ and depth $2h$ that computes $\AND_n$.
\end{theorem}
\begin{proof}
	The inductive circuit construction follows the structure of $\Tt_n^{\boldsymbol{k}}$. Recall that the tree $\Tt_n^{\boldsymbol{k}}$ defines a set $\Bb(\Tt_n^{\boldsymbol{k}})$ of blocks of siblings in the tree. The AND over each such block can be computed via the circuit from \cref{thm:upperbound1}. 
	Below is a schematic visualisation of the top two levels of $\Tt_n^{\boldsymbol{k}}$ where $k_i(n) = n^{1/h}$ for each $i \in [h]$.
	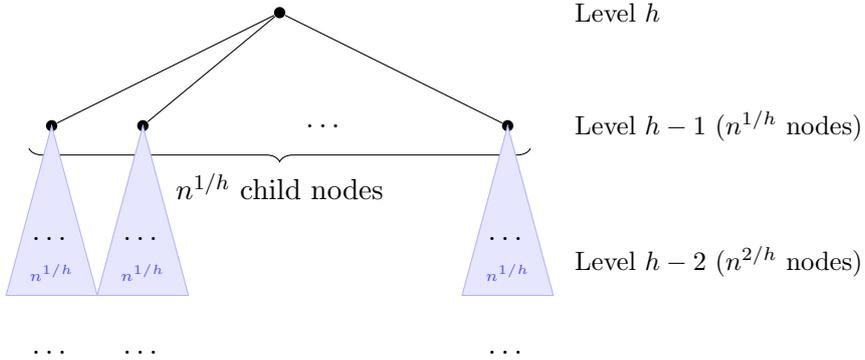
\begin{figure}[h]
		\begin{center}
			\begin{tikzpicture}[
				scale=1.5,
				node/.style={circle, fill=black, inner sep=1.5pt},
				label style/.style={font=\small, color=black}
				]
				
				\node[node] (root) at (0,0) {};
				\node[above=2pt] at (root) {};
				
				\node[node] (L1_1) at (-2,-1) {};
				\node[node] (L1_2) at (-1.2,-1) {};
				\node[node] (L1_10) at (2,-1) {};
				
				\node at (0.4,-1) {$\dots$};
				
				\draw (root) -- (L1_1);
				\draw (root) -- (L1_2);
				\draw (root) -- (L1_10);
				
				\draw [decorate, decoration={brace, amplitude=5pt, mirror}] 
				(-2.2,-1.2) -- (2.2,-1.2) node [black, midway, yshift=-0.5cm] {$n^{1/h}$ child nodes};
				
				\foreach \x in {-2, -1.2, 2} {
					\fill[blue!10, draw=blue!30] (\x,-1) -- (\x-0.4,-2.5) -- (\x+0.4,-2.5) -- cycle;
					\node at (\x, -2) {$\dots$};
					\node[font=\tiny, blue!70] at (\x, -2.3) {$n^{1/h}$};
					\node at (\x, -3) {$\dots$};
				}
				
				\node[anchor=west, label style] at (2.5, 0) {Level $h$};
				\node[anchor=west, label style] at (2.5, -1) {Level $h-1$ ($n^{1/h}$ nodes)};
				\node[anchor=west, label style] at (2.5, -2.2) {Level $h-2$ ($n^{2/h}$ nodes)};
				
			\end{tikzpicture}
			\end{center}
			\caption*{An $n^{1/h}$-ary tree of depth $h$}
	\end{figure}	
	Each blue cone in this picture corresponds to a block $B \in \Bb(\Tt_n^{\boldsymbol{k}})$. On the level of leaves, such a block is a subset of the input variables of size $n^{1/h}$. A block $B$ on a higher level bundles $n^{1/h}$ blocks from the level below. 
	In this example, our circuit will contain one instance of the $\AND_{n^{1/h}}$-circuit from \cref{thm:upperbound1} for each $B \in \Bb(\Tt_n^{\boldsymbol{k}})$. The output of such a circuit will be the AND over $L_0(B)$, that is, the set of all input variables that sit below the block $B$.
	
	Formally, we construct our $\AND_n$-circuit by induction from level $0$ to $h$ of $\Tt_n^{\boldsymbol{k}}$: On level $0$, every block $B \in \Bb(\Tt_n^{\boldsymbol{k}})$ is a subset of leaves. For each such $B$, we invoke \cref{thm:upperbound1} to obtain a $\Sym(B)$-symmetric circuit $C_B$ that computes the $\AND$ over all variables $x_i$ with $i \in B$. 
	
	Next, we consider an arbitrary level $i > 0$ and assume by induction that for all blocks $B \in \Bb(\Tt_n^{\boldsymbol{k}})$ with $B \subseteq L_{i-1}$, a circuit $C_B$ with the following properties has been constructed:\\
	
	\begin{enumerate}
		\item $C_B$ computes the AND over all variables $x_i$ such that $i \in L_0(B)$.
		\item $C_B$ is symmetric under the subgroup of $\Aut(\Tt_n^{\boldsymbol{k}})$ that stabilises $B$ setwise.
	\end{enumerate}	
	~\\
	Now let $B \in \Bb(\Tt_n^{\boldsymbol{k}})$ be a block with $B \subseteq L_i$. To obtain the circuit $C_B$ for this block, we invoke \cref{thm:upperbound1} on the outputs of the circuits $C_{B'}$ for all blocks $B' = B(v)$ for every node $v$ that is a child of some node in $B$. That is, $C_B$ simply computes the AND over the results of all the blocks on level $i-1$ that are bundled in $B$. 
	It is not difficult to check that this circuit $C_B$ again satisfies the two properties above (using the fact that the construction from \cref{thm:upperbound1} is symmetric). 
	 
	For the unique block $B$ on level $h-1$, the circuit $C_B$ is the desired $\Aut(\Tt_n^{\boldsymbol{k}})$-symmetric circuit that computes the $\AND$ over all $n$ input variables.
	The total depth of the construction is $2h$ because the circuit from \cref{thm:upperbound1} has depth $2$, and we use this on $h$ levels. 
	For each block, the subcircuit that \cref{thm:upperbound1} gives us has size at most $2^{\Oo(k_{\text{max}}(n)^{1/r} \cdot \log k_{\text{max}}(n))}$.
	The number of blocks is $|\Bb(\Tt_n^{\boldsymbol{k}})| \leq n$, so the total size of the constructed circuit is at most $n \cdot 2^{\Oo(k_{\text{max}}(n)^{1/r} \cdot \log k_{\text{max}}(n))} = 2^{\Oo(k_{\text{max}}(n)^{1/r} \cdot \log k_{\text{max}}(n) + \log n)}$. The additive $\log n$ term vanishes in the $\Oo$ because $k_{\text{max}}(n) \geq n^{1/h}$.
\end{proof}

\section{Concluding remarks}
Using a clean group-theoretic framework, we have determined the exact size complexity of $\AND_n$ for fully symmetric and nested block-symmetric $\CC^0$-circuits. 
For fully symmetric circuits, it turns out that the depth-$2$ construction from \cite{IdziakKK22LICS} is already optimal. For nested block-symmetric circuits, the optimal size is achieved by recursively nesting that construction. This approach is of course somewhat naive, and we know from \cite[Proposition 4.3]{IdziakKK22LICS} that one can in fact compress its depth down to $h+1$. This is done via a trick that lets the authors chain consecutive $\bbZ_{pq}$-expressions together without explicitly having to compute the $\AND$ over each block. Strangely, the implementation of this trick in \cite[Proposition 4.3]{IdziakKK22LICS} achieves only $1/(r-1)$ in the exponent of the circuit size, rather than $(1/r)$, which we have shown to be optimal for symmetric circuits. After a thorough examination of the depth reduction trick, it seems that this increase in size is perhaps inherent and cannot be avoided if one wishes to achieve a lower depth than $2h$. Thus, our results motivate further efforts to improve the size of the depth-$(h+1)$ construction, or to show that this is impossible (at least under symmetry assumptions). 

Beyond that, we hope that our techniques will provide a basis for further progress towards settling the 30 year old problem $\CC^0$ versus $\ACC^0$. Concretely, we suggest to study the question whether every $\CC^0$-circuit for $\AND_n$ can be efficiently symmetrised. If this is the case, then our symmetric lower bound applies to all of $\CC^0$, and it is separated from $\ACC^0$. If it turns out to be false, then in proving this, we would find a new upper bound construction that achieves a smaller size by breaking symmetries, making progress towards showing $\CC^0 = \ACC^0$. Another promising approach is to try and push the torus polynomial approximation of symmetric $\CC^0$ from \cite{krishan2026cc0lowerboundstorus} further towards weaker symmetry groups, possibly in combination with our framework.

\printbibliography

\newpage

\appendix

\section{Details on symmetry groups and circuits}
\label{sec:appendixGroups}

We first give a detailed proof of the claim that symmetric $\MOD_m$-circuits can always be assumed to be rigid.

\begin{lemma}[Rigidification]
	\label{lem:rigidification}
	Let $\Gamma \leq \Sym_n$.
	If $C$ is a $\Gamma$-symmetric $\MOD_m$-circuit, then there exists a rigid $\Gamma$-symmetric $\MOD_m$-circuit $C'$ that computes the same function as $C$ and satisfies $|C'| \leq |C|$. Moreover,  let $\StabIn \leq \Aut(C)$ be the subgroup consisting of the automorphisms that fix every input gate of $C$.
	There is a surjective map $\delta \colon V(C) \to V(C')$ such that:
	\begin{enumerate}[label=(\alph*)]
		\item For every $g \in V(C')$, $\delta^{-1}(g) \subseteq V(C)$ is an orbit of $\StabIn$.
		\item The map $h \colon \Aut(C) \to \Aut(C')$ that sends every $\sigma \in \Aut(C)$ to the $\sigma' \in \Aut(C')$ with $\sigma'(g) \coloneqq \delta( \sigma(\delta^{-1}(g)))$ for every $g \in V(C')$ is a group homomorphism.
	\end{enumerate}	
\end{lemma}	
\begin{proof}
	We construct $C'$ by merging equivalent gates in $C$. 
	Formally, we define $C'$ and a surjective map $\delta \colon V(C) \to V(C')$ inductively from the input gates of $C$ to the root. 
	By our convention, for every variable, there is a unique input gate with that label in $C$, so input gates never violate rigidity.
	Thus, we let the input gates of $C'$ be the same as in $C$, and define $\delta$ as the identity map on them. 
	Now assume by induction that we have constructed $C'$ and $\delta$ up to layer $d$. We describe the construction on layer $d+1$. Let $V_{d+1} \subseteq V(C)$ be the set of gates in layer $d+1$. For each $\StabIn$-orbit $O \subseteq V_{d+1}$, we introduce a new gate $g_O$ in layer $d+1$ of $C'$, and we let $\delta(g) \coloneqq g_O$ for every $g \in O$. The operation type of $g_O$ is the same as that of each gate in $O$.
	To define the connections between layer $d+1$ and layer $d$ in $C'$,   
	we first note:
	\begin{claim}
		Let $g,g' \in V(C)$ be in the same $\StabIn$-orbit. Then there is a bijection $\gamma \colon gE(C) \to g'E(C)$ such that for each $h \in gE(C)$, $h$ and $\gamma(h)$ are in the same $\StabIn$-orbit.
	\end{claim}
	\begin{claimproof}
		Since $g,g'$ are in the same orbit, there exists $\pi \in \StabIn$ such that $\pi(g) = g'$, and the action of $\pi$ on $gE(C)$ defines a bijection $\gamma \colon gE(C) \to g'E(C)$ with the claimed property.
	\end{claimproof}
	By the claim, for any two gates $g,g' \in O$, it holds that $\{ \delta(h) \mid h \in gE(C) \} = \{ \delta(h) \mid h \in g'E(C)  \}$. Therefore we can pick an arbitrary $g \in O$ and define the set of children of $g_O$ in $C'$ as 
	\[
	g_OE(C') \coloneqq \{ \delta(h) \mid h \in gE(C) \}.
	\]
	For each child $\delta(h)$ of $g_O$ in $C'$, we let the multiplicity of the edge
	between $g_O$ and $\delta(h)$ be defined as follows. Let $m(g,h)$ denote the multiplicity of the edge between $g$ and $h$ in $C$.
	Then in $C'$, the multiplicity of the edge $(g_O, \delta(h))$ is
	\[
	\sum_{h' \in \delta^{-1}(\delta(h)) \cap gE} m(g,h').
	\]
	This finishes the construction of $C'$. Note that by construction, two gates $g_1,g_2 \in V(C)$ are in the same $\StabIn$-orbit if and only if $\delta(g_1) = \delta(g_2)$. Clearly, $|C'| \leq |C|$. 
	
	\begin{claim}
		$C'$ computes the same function as $C$.
	\end{claim}
	\begin{claimproof}
		We show by induction that for every gate $g \in V(C)$, $\delta(g)$ computes the same function as $g$. For the input gates this is clear. Now consider the inductive step for layer $d+1$. Let $g \in V(C)$ be a gate on layer $d+1$, labelled with the operation $\MOD_m^R$, and let $h_1, \dots, h_k$ be its children in $C$. Then it computes 
		\[
		g(\bar{x}) = \begin{cases}
			1 & \text{ if } (\sum_{i \in [k]} m(g,h_i) \cdot h_i(\bar{x}) \mod m) \in R\\
			0 & \text{ otherwise}
		\end{cases}
		\]
		Now the children of $\delta(g)$ in $C'$ are defined as the $\delta$-images of the children of some gate $g^*$ in the same orbit of $g$ that was used in the construction. Hence, there exists a $\pi^* \in \StabIn$ that maps $g$ to $g^*$ and the children of $g$ to the children of $g^*$ (preserving edge multiplicities). It is generally true for every $\pi \in \StabIn$ and any $h \in V(C)$ that $\pi(h)$ and $h$ compute the same function. Thus, we have
		\begin{align*}
			\sum_{i \in [k]} m(g,h_i) \cdot h_i(\bar{x}) &= \sum_{i \in [k]} m(g,h_i) \cdot \pi^*(h_i)(\bar{x}) \\
			&= \sum_{i \in [k]} m(g,h_i) \cdot \delta(\pi^*(h_i))(\bar{x}),
		\end{align*}
		where the last equality holds by induction hypothesis. In the construction of $C'$, the edge multiplicities between $\delta(g^*) = \delta(g)$ and its children $\{\delta(\pi^*(h_i)) \mid i \in [k]\}$ are chosen such that $\delta(g)$ indeed computes the above sum modulo $m$, and checks for membership in $R$. This finishes the inductive step.
	\end{claimproof}

	\begin{claim}
		Property (b) claimed in the lemma holds. 
	\end{claim}
	\begin{claimproof}
		Let $\sigma \in \Aut(C)$. The $\sigma$-image of every $\StabIn$-orbit $O \subseteq V(C)$ is again a $\StabIn$-orbit. Thus, we can define a bijection $\sigma' \colon V(C') \to V(C')$ by setting $\sigma'(g_O) \coloneqq g_{\sigma(O)}$ for every $\StabIn$-orbit $O \subseteq V(C)$. Then $\sigma' = h(\sigma)$ for the map $h \colon \Aut(C) \to \Aut(C')$ defined in property (b).
		By construction of $C'$ and because $\sigma$ is an automorphism of $C$, $\sigma'$ is an automorphism of $C'$, and it is clear that $h$ commutes with the compositional structure of $\Aut(C)$, i.e., it is a group homomorphism. 
	\end{claimproof}
	From this claim it follows that every $\pi \in \Gamma$ extends to a circuit automorphism of $C'$, that is, $C'$ is $\Gamma$-symmetric. Indeed, every $\pi \in \Gamma$ extends to a $\sigma \in \Aut(C)$, and $h(\sigma) \in \Aut(C')$ is then the circuit automorphism of $C'$ that $\pi$ extends to.
	
	The resulting circuit $C'$ is rigid because by construction, every $\StabIn$-orbit of $C'$ is a singleton. The constructed map $\delta \colon V(C) \to V(C')$ satisfies property (a) by construction. 
\end{proof}

\symmetricCircuitSemantics*
\begin{proof}
	By induction on the circuit structure. Let $g$ be an input gate labelled with a variable $x_i$. Let $j \coloneqq \pi(i)$. Then $g' \coloneqq \pi(g)$ is an input gate labelled with $x_j$. It holds $g(\delta(x_1), \dots, \delta(x_n)) = \delta(x_i)$ and $g'(\delta(x_1), \dots, \delta(x_n)) = \delta(x_j)$. In other words, $g(x_1, \dots, x_n)$ is the $i$-th projection, and $g'(x_1, \dots, x_n)$ is the $j$-th projection.
	So, as desired, we have
	\[
	g(\delta(x_1), \dots, \delta(x_n)) = \delta(x_i) = g'(\delta(\pi^{-1}(x_1)), \dots, \delta(\pi^{-1}(x_n))). 
	\]
	For the inductive step, let $g$ be an internal gate of the circuit and assume that the statement holds for all children of $g$. 
	Let $h_1, \dots, h_k$ denote the children of $g$. Let $f$ denote the operation computed by $g$, which is the same as the operation of $\pi(g)$. The proof works for every fully symmetric operation $f$, in particular for $f = \MOD_m^R$. Then 
	\begin{align*}
			g(\delta(x_1), \dots, \delta(x_n)) &= f(h_1(\delta(x_1), \dots, \delta(x_n)), \dots, h_k(\delta(x_1), \dots, \delta(x_n)))\\
		&= f(\pi(h_1)(\delta(\pi^{-1}(x_1)), \dots, \delta(\pi^{-1}(x_n))), \dots, \pi(h_k)(\delta(\pi^{-1}(x_1)), \dots, \delta(\pi^{-1}(x_n))))\\
		&= \pi(g)(\delta(\pi^{-1}(x_1)), \dots, \delta(\pi^{-1}(x_n))).
	\end{align*}	
	The second equality is the induction hypothesis for $h_1, \dots, h_k$. The third equality is true because the gate $\pi(g)$ computes $f$ applied to the outputs of the gates $\pi(h_1), \dots, \pi(h_k)$, and the order of the arguments of $f$ is irrelevant by symmetry. 
\end{proof}

\supportsizes*
\begin{proof}
	The first result is stated in \cite[Theorem 14]{dawar_symmetric_2024} for Boolean circuits symmetric under the action of $\Sym_n$ on variables $\Xx \coloneqq \{x_{ij} \mid i,j \in [n]\}$. The symmetric circuits we consider here can be viewed as circuits in the variables $\{x_{ii} \mid i \in [n]\} \subseteq \Xx$, and thus, \cite[Theorem 14]{dawar_symmetric_2024} also applies here (the particular operation type of the gates is irrelevant for this).
	
	With some more work, (2) follows from (1). Fix a gate $g \in V(C)$. We can assume that $g$ is an internal gate, as for input gates, there always exists a $B$-support of size $0$ or $1$, depending whether the index of the variable that $g$ is labelled with is in $B$ or not.  
	Let $B \in \Bb(\Tt_n^{\boldsymbol{k}})$. Define a new circuit $C_B$ from $C$ as follows.	 
	Remove all input variables $x_i$ such that $i \notin L_0(B)$. Then, for every $v \in B$, identify all input gates labelled with variables $x_i$, for every $i \in L_0(v)$. This leaves us with a circuit with one input variable for every $v \in B$. Call this circuit $C'_B$. The modifications we applied to $C$ may have destroyed the rigidity. Thus, we rigidify $C'_B$ again using \cref{lem:rigidification}. Let $C_B$ be the resulting circuit. The original circuit $C$ is in particular symmetric under $\Stab(B)$, the setwise stabiliser of $B$ in $\Aut(\Tt_n^{\boldsymbol{k}})$; hence $C'_B$ and thus $C_B$ is also $\Stab(B)$-symmetric. The action of $\Stab(B)$ on $B$ is that of $\Sym(B)$, so $C_B$ can really be seen as a $\Sym(B)$-symmetric circuit.
	By the assumption on orbit size in (2), $\maxorb_{\Sym(B)}(C_B) \leq \binom{|B|}{k_B}$, and we have assumed that $1 \leq k_B \leq \frac{|B|}{4}$. Let $\delta \colon V(C'_B) \to V(C_B)$ be the map from \cref{lem:rigidification}. We are assuming that $|B| > 8$. Hence, (1) can be applied to $C_B$, where we just rename $B$ to $[n]$. This means that in $C_B$, the gate $\delta(g)$ has a support $S \subseteq B$ of size at most $k_B$. We now show that this is also a $B$-support of $g$ in $C$.
	We have to show that \[
	\StabP_{\Sym(B)}(S) \leq \Stab_{\Aut(\Tt_n^{\boldsymbol{k}})}(g)|_B. 
	\]
	So let $\pi \in\StabP_{\Sym(B)}(S)$.
	Let $\Delta \leq \Aut(\Tt_n^{\boldsymbol{k}})$ be the subgroup consisting of all $\sigma \in \Aut(\Tt_n^{\boldsymbol{k}})$ that fix $B$ setwise and act like $\pi$ on $B$. Since $C$ is $\Aut(\Tt_n^{\boldsymbol{k}})$-symmetric, every $\sigma \in \Delta$ extends to a circuit automorphism $\theta^{\sigma}_C$ of $C$. This also induces a circuit automorphism $\theta^{\sigma}_{C'_B} \in \Aut(C'_B)$, which behaves like $\theta^{\sigma}_C$ on the internal gates (noting that except for the input gates and the connections to them, $C$ and $C'_B$ are identical circuits). \cref{lem:rigidification} supplies a group homomorphism $h \colon \Aut(C'_B) \to \Aut(C_B)$, so $\theta^{\sigma}_{C_B} \coloneqq h(\theta^{\sigma}_{C'_B}) \in \Aut(C_B)$. 
	Now by definition of support, $\theta^{\sigma}_{C_B}$ fixes $\delta(g)$: Indeed, $\theta^{\sigma}_{C_B}$ acts on the inputs of $C_B$ as $\pi$, and $\pi$ fixes the support $S$ of $\delta(g)$ pointwise. 
	Therefore, by \cref{lem:rigidification} (b), $\theta^{\sigma}_{C'_B}$ fixes $\delta^{-1}(\delta(g)) \subseteq V(C'_B)$ setwise. This is an orbit of the group $\StabIn \leq \Aut(C'_B)$ that fixes every input gate of $C'_B$.
	We can view $\theta^{\sigma}_{C'_B}$ as an element of a coset of $\StabIn$ in $\Aut(C'_B)$, say $\theta^{\sigma}_{C'_B} \in \gamma \StabIn$ for some $\gamma \in \Aut(C'_B)$. 
	Because of the way that $C'_B$ arises from the rigid circuit $C$, every element of $\StabIn$ corresponds to some automorphism of $C$ that extends a permutation in $\Aut(\Tt_n^{\boldsymbol{k}})$ which fixes every element of $B$. Hence, the coset $\gamma \StabIn$ depends only on the choice of $\pi \in \StabP_{\Sym(B)}(S)$ that we fixed.
	Therefore, there must be a $\sigma \in \Delta$ such that $\theta^{\sigma}_{C'_B} \in \gamma \StabIn$ not only fixes $\delta^{-1}(\delta(g))$ setwise, but also fixes the gate $g$ itself. 
	Then for this $\sigma$, the automorphism $\theta^{\sigma}_C$ of $C$ also fixes the gate $g$ in $C$. 
	
	 Hence, this particular $\sigma \in \Delta \leq \Aut(\Tt_n^{\boldsymbol{k}})$ is in $\Stab_{\Aut(\Tt_n^{\boldsymbol{k}})}(g)$, and its restriction to $B$ is $\pi$. Thus, $\pi \in \Stab_{\Aut(\Tt_n^{\boldsymbol{k}})}(g)|_B$, which is what we had to show.
\end{proof}

\end{document}